\documentclass{new_tlp}
\usepackage{mathptmx}
\usepackage{xspace}
\usepackage{amsmath}
\usepackage{amsfonts}
\usepackage{amssymb}
\let\origtheorem\theorem
\let\origproof\proof
\usepackage{z-eves}
\let\theorem\origtheorem
\let\proof\origproof
\usepackage{algorithm}
\usepackage{algpseudocode}
\usepackage{enumitem}
\usepackage{xcolor}
\usepackage{soul}
\usepackage{comment}
\usepackage{fancyvrb}
\usepackage{url}

\newcommand{\true}{\mathit{true}}
\newcommand{\false}{false}
\newcommand{\lfun}{\longrightarrow}
\newcommand{\qed}{\hfill$\Box$}

\newcommand{\card}[1]{\lvert #1 \rvert}
\newcommand{\e}{\emptyset}
\renewcommand{\plus}{\mathbin{\scriptstyle\sqcup}}
\newcommand{\w}{\{\cdot \plus \cdot\}}
\newcommand{\disj}{\parallel}
\newcommand{\Ndisj}{\not\disj}
\renewcommand{\Cup}{un}
\renewcommand{\Cap}{inters}
\newcommand{\Ncup}{nun}
\newcommand{\Size}{size}

\newcommand{\Diff}{diff}

\newcommand{\CARD}{\card{\cdot}}
\newcommand{\setlog}{$\{log\}$\xspace}

\newcommand{\CLPSET}{CLP($\mathit{SET}$)\xspace}

\newcommand{\LCARD}{\mathcal{L}_{\CARD}}
\newcommand{\SATCARD}{\mathit{SAT}_{\CARD}}

\newcommand{\LZa}{\mathcal{L}_{Za}}
\newcommand{\Zarba}{\mathit{SAT}_{Za}}

\newcommand{\STEPCARD}{\mathsf{solve\_size}}

\newcommand{\Var}{\mathcal{V}}
\newcommand{\Set}{\mathsf{S}}
\newcommand{\Int}{\mathsf{Z}}
\newcommand{\FSet}{\mathcal{F}_\Set}
\newcommand{\PiSet}{\Pi_\mathsf{S}}
\newcommand{\Ur}{\mathsf{U}}
\newcommand{\FUr}{\mathcal{F}_\Ur}
\newcommand{\FInt}{\mathcal{F}_\Int}
\newcommand{\PiInt}{\Pi_\Int}
\newcommand{\sSet}{\mathsf{Set}}
\newcommand{\sInt}{\mathsf{Int}}
\newcommand{\sUr}{\mathsf{Ur}}

\newcommand{\TCARD}{\mathcal{T}_{\CARD}}
\newcommand{\TInt}{\mathcal{T}_\Int}
\newcommand{\TUr}{\mathcal{T}_\Ur}
\newcommand{\CCARD}{\mathcal{C}_{\CARD}}

\newcommand{\map}{\mathcal{Z}}

\newcommand{\iS}{\mathcal{R}}
\newcommand{\iF}[1]{(#1)^\iS}

\newcommand{\q}{\text{\normalfont\'{}}}
\newcommand{\ql}{\hspace{5pt}\text{\normalfont\'{}}}
\newcommand{\qr}{\text{\normalfont\'{}}\hspace{5pt}}
\newcommand{\why}[1]{\tag*{{\footnotesize [by #1]}}}
\newcommand{\by}[1]{{\footnotesize [by #1]}}

\newtheorem{definition}{Definition}
\newtheorem{example}{Example}
\newtheorem{remark}{Remark}
\newtheorem{theorem}{Theorem}
\newtheorem{lemma}{Lemma}

\title[Integrating Cardinality Constraints into CLP with Sets]{Integrating Cardinality Constraints into Constraint Logic Programming with Sets}

\author[M. Cristi\'a and G. Rossi]
{MAXIMILIANO CRISTI\'A \\
Universidad Nacional de Rosario \\
Argentina    \\
E-mail: cristia@cifasis-conicet.gov.ar
\and
GIANFRANCO ROSSI \\
Universit\`a di Parma \\
Italy    \\
E-mail: gianfranco.rossi@unipr.it
}

\jdate{n/a}
\pubyear{n/a}
\pagerange{\pageref{firstpage}--\pageref{lastpage}}
\doi{n/a}

\begin{document}
\label{firstpage}

\maketitle

\begin{abstract}
Formal reasoning about finite sets and cardinality is important for many
applications, including software verification, where very often one needs to
reason about the size of a given data structure. The Constraint Logic
Programming tool \setlog provides a decision procedure for deciding the
satisfiability of formulas involving very general forms of finite sets,
although it does not provide cardinality constraints. In this paper we adapt
and integrate a decision procedure for a theory of finite sets with cardinality
into \setlog. The proposed solver is proved to be a decision procedure for its
formulas. Besides, the new CLP instance is implemented as part of the \setlog
tool. In turn, the implementation uses Howe and King's Prolog SAT solver and
Prolog's CLP(Q) library, as an integer linear programming solver. The empirical
evaluation of this implementation based on +250 real verification conditions
shows that it can be useful in practice.

\medskip\noindent
\emph{Under consideration in Theory and Practice of Logic Programming (TPLP)}
\end{abstract}



\section{Introduction}

Set theory is a well-established vehicle for formal modeling, specification,
analysis and verification of software systems. Formal notations such as B
\cite{Abrial00} and Z \cite{Spivey00} and tools such as ProB \cite{Leuschel00},
Atelier-B \cite{atelierb} and Z/EVES \cite{Saaltink01} are good examples of
that claim. Hence, it is important to extend the capabilities of existing tools
and develop new ones for set theory as applied in the context of verification.
Besides, when these methods and tools are used for formal verification and
analysis, it is necessary to discharge a number of verification conditions or
proof obligations. Then, tools capable of automating such proofs are essential
to render the development process cost-effective. Decision procedures play a
key role in proof automation. Indeed, if a decision procedure exists for a
fragment of set theory, then it would be possible to automate the proofs of
verification conditions lying in this fragment.

\setlog (read `setlog') \cite{DBLP:journals/jlp/DovierOPR96,setlog} is a Constraint Logic
Programming (CLP) language and satisfiability solver implemented in Prolog
providing: \emph{i)} a decision procedure for the algebra of \emph{hereditarily
finite sets}, i.e., finitely nested sets that are finite at each level of
nesting \cite{Dovier00}; \emph{ii)} a decision
procedure for a very expressive fragment of the class of finite set relation
algebras (\citeANP{DBLP:journals/jar/CristiaR20} \citeyearNP{DBLP:journals/jar/CristiaR20,DBLP:conf/RelMiCS/CristiaR18}); and
\emph{iii)} a decision procedure for restricted intensional sets (RIS) (\citeANP{DBLP:conf/cade/CristiaR17} \citeyearNP{DBLP:journals/jar/CristiaR21a,DBLP:conf/cade/CristiaR17}). Several
in-depth empirical evaluations provide evidence that \setlog is able to solve
non-trivial problems
(\citeANP{DBLP:journals/jar/CristiaR20} \citeyearNP{DBLP:journals/jar/CristiaR21a,DBLP:journals/jar/CristiaR20,DBLP:conf/RelMiCS/CristiaR18,DBLP:conf/cade/CristiaR17,CristiaRossiSEFM13}),
in particular as an automated verifier of security properties
(\citeANP{DBLP:journals/jar/CristiaR21} \citeyearNP{DBLP:journals/jar/CristiaR21,Cristia2021}). All of these decision
procedures are based on the notion of \emph{set unification} \cite{Dovier2006}.

In this paper we add to \setlog a decision procedure for the algebra of finite
sets extended with cardinality constraints. This extension is important in
terms of formal software verification because there are situations where we
need to reason about the size of a given data structure and not only about what
its elements are. For example, within the algebra of finite sets one can
partition a given set into two disjoint subsets, $C = A \cup B \land A \cap B =
\emptyset$, but there is no way to state that $A$ and $B$ must be of the same
cardinality. In practice these constraints might appear, for instance, when
part of a data container must be put into a cache---a simple \setlog program is
shown in \ref{app:datacontainer}. Specifically, cardinality constraints appear
in the verification of some distributed algorithms
\cite{DBLP:conf/cav/BerkovitsLLPS19,Alberti2017} and are at the base of the
notions of integer interval, array and list.

At an abstract level, the new decision procedure combines the decision
procedure for the algebra of finite sets already existing in \setlog with a
decision procedure for sets with cardinality constraints proposed by
\citeN{DBLP:conf/frocos/Zarba02}.  Zarba proves that a theory of finite sets
equipped with the classic set theoretic operators, including cardinality,
combined with linear integer constraints is decidable. In his work, Zarba is
interested in proving a decidability result; as far as we know Zarba's
algorithm has never been implemented before. In fact, the new decision
procedure first uses all the power of \setlog to produce a simplified,
equivalent formula that can be passed to Zarba's algorithm which makes a final
judgment about its satisfiability, in case it contains cardinality constraints.
In this way, \setlog performs as well as before on the class of formulas it was able to
deal with previously.

As a consequence of the fact that the new decision procedure is still based on
set unification, it can deal with sets of sets nested at any depth. For
example, the decision procedure is able to give all possible
solutions for a goal such as $\card{\{\{x\},\{y,z\}\}} = n$, where $x$, $y$,
$z$ and $n$ are variables.

Zarba's algorithm is implemented by integrating the Prolog Boolean SAT solver
developed by \citeN{DBLP:journals/tcs/HoweK12} with SWI-Prolog's implementation
of the CLP(Q) system \cite{holzbaur1995ofai}. As a result the implementation
integrates three Prolog-based systems: Howe and King's SAT solver, CLP(Q) and
\setlog.

Solving formulas over a theory of sets and cardinality is not new
\cite{DBLP:conf/cade/FerroOS80,Gervet01}. However, our proposal clearly
distinguishes itself from all previous works in some aspects that constitute
the main contributions of this paper: \emph{a)} our implementation is deeply
rooted in the CLP framework and thus inherits all its properties; in
particular, \setlog preserves its features as a CLP language and as a
satisfiability solver; \emph{b}) our CLP system produces a finite
representation of all possible solutions of any satisfiable formula of its
input language; \emph{c}) as the decision procedure is based on set unification
it handles set elements of any kind including nested sets; and \emph{d)} this
is the first implementation of Zarba's algorithm and it is shown to perform
better than some other systems.

\paragraph{Structure of the paper.}
Section \ref{language} presents the syntax and semantics of the constraint
language for finite sets with cardinality constraints. The overall structure of
the constraint solver for that language is introduced in Section \ref{satcard}.
The main routine dealing with cardinality constraints is presented in Section
\ref{sizesolver}, where we also include a description of Zarba's algorithm. In
Section \ref{decproc} we prove that the resulting solver is indeed a decision
procedure for our language. Besides deciding the satisfiability of cardinality
formulas, the solver is able to find a particular form of their solutions, as
we explain in Section \ref{minimal}. Section \ref{impl} shows how \setlog works
with cardinality constraints, in particular in the context of formal
verification (Section \ref{formver}); an empirical evaluation is also reported
(Section \ref{empirical}). We compare our approach with others in Section
\ref{related}. Some concluding remarks are provided in Section \ref{concl}.

\section{\label{language}$\LCARD$: a language for finite sets and cardinality}

In this section we describe the syntax and semantics of our set-based language
$\LCARD$ (read `l-card'). This is a quantifier-free first-order predicate
language with three distinct sorts: the sort $\sSet$ of all terms denoting
sets, the sort $\sInt$ of terms denoting integer numbers, and the sort $\sUr$
of all other terms. Terms of each sort are allowed to enter in the formation of
set terms (in this sense, the designated sets are hybrid), no nesting
restrictions being enforced (in particular, membership chains of any finite
length can be modeled). A handful of reserved predicate symbols endowed with a
pre-designated set-theoretic meaning is available. The usual linear integer
arithmetic operators are available as well. Formulas are built in the usual way
by using conjunction and disjunction. A few more complex operators (in the form
of predicates) are defined as $\LCARD$ formulas, thus making it simpler for the
user to write complex formulas.

\subsection{Syntax}\label{syntax}

The syntax of the language is defined primarily by giving the signature upon
which terms and formulas are built.

\begin{definition}[Signature]\label{signature}
The signature $\Sigma_{\CARD}$ of $\LCARD$ is a triple $\langle
\mathcal{F},\Pi,\Var\rangle$ where:
\begin{itemize}
\item $\mathcal{F}$ is the set of constants and function symbols along with their sorts, partitioned as
      $\mathcal{F} \defs \FSet \uplus \FInt \uplus \FUr$, where
      $\FSet \defs \{\e,\plus\}$, $\FInt = \{0,-1,1,-2,2,\dots\} \cup \{+,-,*\}$
      and $\FUr$ is a set of uninterpreted constant and function symbols.
\item $\Pi$ is the set of predicate symbols along with their sorts, partitioned as
$\Pi \defs \Pi_{=} \cup  \PiSet \cup \Pi_{\Size} \cup \PiInt$, where $\Pi_{=}
\defs  \{=,\neq\}$, $\PiSet \defs \{\in,\notin,\Cup,\disj\}$, $\Pi_{\Size}
\defs \{\Size\}$, and  $\PiInt \defs \{\leq\}$.
%
\item $\Var$ is a denumerable set of variables partitioned as
$\Var \defs \Var_\Set \cup \Var_\Int \cup \Var_\Ur$.
\qed
\end{itemize}
\end{definition}

Intuitively, $\e$ represents the empty set; $\{x \plus A\}$ represents the
set\footnote{$\plus$ is akin to Prolog's list constructor
`$\mid$'.}\textsuperscript{-}\footnote{In \setlog, $\e$ is written as \{\}
and $\plus$ as /, see Section \ref{impl}.} $\{x\} \cup A$; and $\Var_\Set$,
$\Var_\Int$ and $\Var_\Ur$ represent sets of variables ranging over sets,
integers and ur-elements\footnote{Ur-elements (also known as atoms or
individuals) are objects which have no elements but are distinct from the empty
set.}, respectively.

Sorts of function and predicate symbols are specified as follows: if $f$
(resp., $\pi$) is a function (resp., a predicate) symbol of arity $n$,
then its sort is an $n+1$-tuple $\langle s_1, \ldots ,s_{n+1} \rangle$ (resp.,
an $n$-tuple $\langle s_1, \ldots ,s_n \rangle$) of non-empty subsets of the
set $\{ \sSet , \sInt, \sUr \}$ of sorts. This notion is denoted by $f:\langle
s_1, \ldots ,s_{n+1}\rangle$ (resp., by $\pi:\langle s_1, \ldots ,s_n\rangle
$). Specifically, the sorts of the elements of $\mathcal{F}$ and $\Var$ are the
following.

\begin{definition}[Sorts of function symbols and variables]\label{d:sorts}
The sorts of the symbols in $\mathcal{F}$ are as follows:
\begin{flalign*}
  \quad\quad & \e: \langle \{\sSet \} \rangle & \\
 & \mathsf{\w: \langle \{\sSet , \sInt, \sUr\}, \{ \sSet \} , \{ \sSet\}\rangle } & \\
 & c: \langle \{\sInt \} \rangle \text{, for any $c \in \{0,-1,1,-2,2,\dots\}$} & \\
 & \cdot + \cdot, \cdot - \cdot, \cdot * \cdot:
  \langle \{\sInt\}, \{\sInt\} , \{\sInt\}\rangle & \\
 & f: \langle \underbrace{\{\sSet,\sInt,\sUr\}, \ldots,
           \{\sSet,\sInt,\sUr\}}_n, \{{\sf
\sUr}\}\rangle\text{, if $f \in \FUr$ is of arity $n \ge 0$}. &
 \end{flalign*}
The sorts of variables are as follows:
\begin{flalign*}
  \quad\quad & v: \langle \{\sSet \} \rangle \text{, if $v \in \Var_\Set$} & \\
 & v: \langle \{\sInt \} \rangle \text{, if $v \in \Var_\Int$} & \\
 & v: \langle \{\sUr \} \rangle \text{, if $v \in \Var_\Ur$} & \tag*{\qed}
 \end{flalign*}
\end{definition}

\begin{definition}[Sorts of predicate symbols]\label{d:sorts_pred}
The sorts of the predicate symbols in $\Pi$ are as follows
(symbols $\Cup$ and $\Size$ are prefix; all other symbols in $\Pi$ are infix):
\begin{flalign*}
 \quad\quad & =,\neq: \langle \{\sSet , \sInt, \sUr \}, \{ \sSet , \sInt, \sUr \} \rangle  & \\
 & \in,\notin: \langle \{\sSet, \sInt, \sUr \} , \{\sSet \} \rangle & \\
 & \Cup: \langle \{\sSet \} , \{\sSet \}, \{\sSet \} \rangle & \\
 & \disj: \langle \{\sSet \} , \{\sSet \} \rangle & \\
 & \Size: \langle \{\sSet \} , \{\sInt \} \rangle & \\
 & \leq: \langle \{\sInt \} , \{\sInt \} \rangle & \tag*{\qed}
 \end{flalign*}
\end{definition}

Note that arguments of $=$ and $\neq$ can be of any of the three considered
sorts. We do not have distinct symbols for different sorts, but the
interpretation of $=$ and $\neq$ (see Section \ref{semantics}) depends on the
sorts of their arguments.

The set of admissible (i.e., well-sorted) $\LCARD$ terms is defined as follows.

\begin{definition}[$\CARD$-terms]\label{LCARD-terms}
The set of \emph{$\CARD$-terms}, denoted by $\TCARD$, is the minimal subset of
the set of $\Sigma_{\CARD}$-terms generated by the following grammar complying
with the sorts as given in Definition \ref{d:sorts}:
\begin{flalign*}
 \quad\quad C     &::=   0 \hspace{2pt}|\hspace{2pt}
          {-1} \hspace{2pt}|\hspace{2pt}
            1 \hspace{2pt}|\hspace{2pt}
          {-2} \hspace{2pt}|\hspace{2pt}
            2 \hspace{2pt}|\hspace{2pt}
           \dots & \\
\TInt &::= C \hspace{2pt}|\hspace{2pt}
          \Var_\Int \hspace{2pt}|\hspace{2pt}
          C * \Var_\Int \hspace{2pt}|\hspace{2pt}
          \Var_\Int * C \hspace{2pt}|\hspace{2pt}
          \TInt + \TInt \hspace{2pt}|\hspace{2pt}
          \TInt - \TInt & \\
\TCARD & ::=
  \TInt \hspace{2pt}|\hspace{2pt}
  \TUr \hspace{2pt}|\hspace{2pt} \Var_\Ur \hspace{2pt}|\hspace{2pt}
  \mathit{Set} & \\
\mathit{Set} & ::=
   \q\e\qr
   \hspace{2pt}|\hspace{2pt}
   \Var_S
   \hspace{2pt}|\hspace{2pt}
      \q\{\qr \TCARD
           \ql\hspace{-2pt}\plus\hspace{-2pt}\qr \mathit{Set} \ql\}\q &
\end{flalign*}
where $\TInt$ (resp., $\TUr$) represents any non-variable $\FInt$-term
(resp., $\FUr$-term). \qed
\end{definition}

As can be seen, through rules $C$ and $\TInt$, the grammar allows only integer
linear terms.

If $t$ is a term $f(t_1,\dots,t_n)$, $f \in \mathcal{F}, n \ge 0$, and $\langle
s_1, \ldots ,s_{n+1} \rangle$ is the sort of $f$, then we say that $t$ is of
sort $\langle s_{n+1} \rangle$. The sort of any $\CARD$-term $t$ is always
$\langle \{\sSet\}\rangle$ or $\langle \{\sInt\} \rangle$ or $\langle \{\sUr\}
\rangle$. For the sake of simplicity, we simply say that $t$ is of sort $\sSet$
or $\sInt$ or $\sUr$, respectively. In particular, we say that a $\CARD$-term
of sort $\sSet$ is a {\em set term}, and that set terms of the form  $\{t_1
\plus t_2\}$ are \emph{extensional} set terms. The first parameter of an
extensional set term is called \emph{element part} and the second is called
\emph{set part}. Observe that one can write terms representing sets which are
nested at any level.

Hereafter, we will use the following notation for extensional set terms:
$\{t_1,t_2,\dots,t_n \plus t\}$, $n \ge 1$, is a shorthand for $\{t_1
\plus \{t_2 \,\plus\, \cdots \{ t_n \plus t\}\cdots\}\}$, while
$\{t_1,t_2,\dots,t_n\}$ is a shorthand for $\{t_1,t_2,\dots,t_n \plus \e\}$.
Moreover, we will use the following naming conventions: $A, B, C, D$
stand for terms of sort $\sSet$;  $i, j, k, m$
stand for terms of sort $\sInt$; $a, b, c, d$ stand for terms of sort $\sUr$; and $x, y, z$
stand for terms of any of the three sorts.

\begin{example}[Set terms]
The following $\Sigma_{\CARD}$-terms are set terms:
\begin{flalign*}
  \quad\quad & \e &\\
 & \{x \plus A\} & \\
 & \{4+k,f(a,b)\}, \text{ i.e., } \{4+k \plus \{f(a,b) \plus \e \}\},
  \text{ where $f$ is a (uninterpreted) symbol in $\FUr$.} &
\end{flalign*}
On the opposite, $\{x \plus 17\}$ is not a set term. \qed
\end{example}


The sets of well-sorted $\LCARD$ constraints and
formulas are defined as follows.

\begin{definition}[$\card{\cdot}$-constraints]\label{primitive-constraint}
If $\pi \in \Pi$ is a predicate symbol of sort $\langle s_1, \ldots , s_n
\rangle$, and for each $i=1,\ldots , n$, $t_i$ is a $\CARD$-term of sort
$\langle s'_i \rangle$ with $s'_i \subseteq s_i$, then $\pi (t_1,\ldots ,t_n)$
is a \emph{$\CARD$-constraint}. The set of $\CARD$-constraints is denoted by
$\CCARD$. \qed
\end{definition}

$\CARD$-constraints whose arguments are of sort $\sSet$ (including $\Size$
constraints) will be called \emph{set constraints}; $\CARD$-constraints whose
arguments are of sort $\sInt$ will be called \emph{integer constraints}.

\newcommand{\FCARD}{\Phi_{\CARD}}

\begin{definition}[$\CARD$-formulas]\label{formula}
The set of $\CARD$-formulas, denoted by $\FCARD$, is given
by the following grammar:
\begin{flalign*}
 \quad\quad & \FCARD ::=
  \true \mid \false \mid \CCARD \mid \FCARD \land \FCARD \mid \FCARD \lor
  \FCARD &
\end{flalign*}
where $\CCARD$ represents any element belonging to the set of
$\CARD$-constraints. \qed
\end{definition}

\begin{example}[$\CARD$-formulas]\label{ex:formulas}
The following are $\CARD$-formulas:
\begin{flalign*}
 \quad\quad & a \in A \land a \notin B \land \Cup(A,B,C) \land C = \{x \plus D\} & \\
 & \Cup(A,B,C) \land n+k > 5 \land \Size(C,n) \land B \neq \e & \\
 & x \in A \land B \in A \land \Size(A,x) \land \Size(B,y) \land x<y &
\end{flalign*}
On the contrary, $\Cup(A,B,23)$ is not a $\CARD$-formula because $\Cup(A,B,23)$
is not a $\CARD$-constraint ($23$ is not of sort $\sSet$ as required by the
sort of $\Cup$). \qed
\end{example}

As we will show in Section \ref{expressiveness}, the language does not need a primitive negation connective, thanks to the presence of negative constraints.

\subsection{\label{semantics}Semantics}

Sorts and symbols in $\Sigma_{\CARD}$ are interpreted according to the
interpretation structure $\iS \defs \langle D,\iF{\cdot}\rangle$, where $D$ and
$\iF{\cdot}$ are defined as follows.

\begin{definition} [Interpretation domain] \label{def:int_dom}
The interpretation domain $D$ is partitioned as $D \defs D_\sSet \cup D_\sInt \cup D_\sUr$
where:
\begin{itemize}
\item $D_\sSet$ is the set of all hereditarily finite hybrid
sets built from elements in $D$. Hereditarily finite sets are those sets that
admit (hereditarily finite) sets as their elements, that is sets of sets.
\item $D_\sInt$ is the set of integer numbers, $\mathbb{Z}$.
\item $D_\sUr$ is a collection of other objects. \qed
\end{itemize}
\end{definition}

\begin{definition} [Interpretation function] \label{app:def:int_funct}
The interpretation function $\iF{\cdot}$ is defined as follows:
\begin{itemize}
\item Each sort $\mathsf{X} \in \{\sSet,\sInt,\sUr\}$ is mapped to
      the domain $D_\mathsf{X}$.

\item For each sort $\mathsf{X}$, each variable $x$ of sort $\mathsf{X}$ is mapped to
      an element $x^\iS$ in  $D_\mathsf{X}$.

\item The constant and function symbols in $\mathcal{F}_\Set$ are
interpreted as follows:
  \begin{itemize}
  \item $\e$ is interpreted as the empty set, namely $\e^\iS = \e$
  \item $\{ x \plus A \}$ is interpreted as the set $\{x^\iS\} \cup A^\iS$.
  \end{itemize}

\item The constant and function symbols in $\mathcal{F}_\Int$ are
interpreted as follows:
\begin{itemize}
\item Each element of \{0,-1,1,-2,2,\dots\} is interpreted as the corresponding integer number
\item $i + j$ is interpreted as $i^\iS + j^\iS$
\item $i - j$ is interpreted as $i^\iS - j^\iS$
\item $i * j$ is interpreted as $i^\iS * j^\iS$
\end{itemize}

\item The predicate symbols in $\Pi$ are interpreted as follows:
  \begin{itemize}
   \item $x = y$, where $x$ and $y$ have the same sort $\mathsf{X}$, is interpreted as the
   identity between $x^\iS$ and $y^\iS$ in $D_\mathsf{X}$;
   otherwise, $x = y$ is interpreted as being $\false$
   \item $x \in A$ is interpreted as $x^\iS \in A^\iS$
   \item $\Cup(A,B,C)$ is interpreted as $C^\iS = A^\iS \cup B^\iS$
   \item $A \disj B$ is interpreted as $A^\iS \cap B^\iS = \emptyset$
   \item $\Size(A,k)$ is interpreted as $\card{A^\iS} = k^\iS$
   \item $i \leq j$ is interpreted as $i^\iS \leq j^\iS$
   \item $x \neq y$ and $x \notin A$ are
   interpreted as $\lnot x = y$ and $\lnot x \in A$, respectively.
\qed
\end{itemize}
\end{itemize}
\end{definition}

It is worth noting that $\Size(A,k)$ is interpreted as the
\emph{cardinality}, i.e., the number of elements, of the set denoted by $A$,
and it is not to be confused with the term size, i.e., the number of function
symbols appearing in the term $A$.

The interpretation structure $\iS$ is used to evaluate each $\CARD$-formula
$\Phi$ into a truth value $\Phi^\iS = \{\true,\false\}$ in the following way:
set constraints (resp., integer constraints) are evaluated by $\iF{\cdot}$
according to the meaning of the corresponding predicates in set theory (resp.,
in number theory) as defined above; $\CARD$-formulas are evaluated by
$\iF{\cdot}$ according to the rules of propositional logic. A $\LCARD$-formula
$\Phi$ is \emph{satisfiable} iff there exists an assignment $\sigma$ of values
from ${\cal D}$ to the variables of $\Phi$, respecting the sorts of the
variables, such that $\Phi[\sigma]$ is true in $\iS$, i.e., $\iS \models
\Phi[\sigma]$. In this case, we say that $\sigma$ is a \emph{successful
valuation} (or, simply, a \emph{solution}) of $\Phi$.

In particular, observe that equality between two set terms is interpreted as
the equality in $D_\sSet$; that is, as set equality between hereditarily finite
hybrid sets. Such equality is regulated by the standard \emph{extensionality
axiom}, which has been proved to be equivalent, for hereditarily finite sets,
to the following equational axioms \cite{Dovier00}:
\begin{gather}
\{x, x \plus A\} = \{x \plus A\} \tag{$Ab$} \label{Ab} \\
\{x, y \plus A\} = \{y, x \plus A\} \tag{$C\ell$} \label{Cl}
\end{gather}
Axiom \eqref{Ab} states that duplicates in a set term do not matter
(\emph{Absorption property}). Axiom \eqref{Cl} states that the order of
elements in a set term is irrelevant (\emph{Commutativity on the left}). These
two properties capture the intuitive idea that, for instance, the set terms
$\{1,2\}$, $\{2,1\}$, and $\{1,2,1\}$ all denote the same set.

\subsection{\label{expressiveness}{Derived Constraints}}

$\LCARD$ can be extended to support other set and integer operators definable
by means of suitable $\LCARD$ formulas.

\citeN{Dovier00} proved that the collection of predicate symbols
in $\Pi_{=} \cup \PiSet$ is
sufficient to define constraints implementing the set operators $\cap$,
$\subseteq$ and $\setminus$. For example, $A \subseteq B$ can be defined by the
$\LCARD$ formula $\Cup(A,B,B)$. Likewise, $\{=,\neq\} \cup \PiInt$
is sufficient to define $<$, $>$ and $\geq$. With a slight abuse of
terminology, we say that the set  and integer predicates that are specified by
$\CARD$-formulas are \emph{derived constraints}.

Whenever a formula contains a derived constraint, the constraint is replaced by
its definition turning the given formula into an $\LCARD$ formula. Precisely, if
formula $\Phi$ is the definition of constraint $c$, then $c$ is replaced by
$\Phi$ and the solver checks satisfiability of $\Phi$ to determine
satisfiability of $c$. Thus, we can completely ignore the presence of derived
constraints in the subsequent discussion about constraint solving and formal
properties of our solver.


The negated versions of set and integer operators can be introduced as derived
constraints, as well. The derived constraint for $\lnot\cup$ and
$\lnot\disj$ (called $\Ncup$ and $\Ndisj$, respectively) are shown in
\cite{Dovier00}.
For example, $\lnot(A \cup B = C)$ is introduced as:
\begin{equation}\label{e:nun}
\Ncup(A,B,C) \defs
     (n \in C \land n \notin A \land n \notin B)
     \lor (n \in A \land n \notin C)
     \lor (n \in B \land n \notin C)
\end{equation}
With a
little abuse of terminology, we will refer to these predicates as
\emph{negative constraints}.

Thanks to the availability of negative constraints, (general) logical negation
is not strictly necessary in $\LCARD$.

Now that we have derived and negative constraints it is easy to see that
$\LCARD$ expresses the Boolean algebra of sets with cardinality.


\begin{remark}[\CLPSET]
\setlog
provides an implementation of the CLP instance \CLPSET \cite{Dovier00}. In
turn, \CLPSET is based on a constraint language including
$\mathcal{F}_\Set$ and $\PiSet$, with the same sorts; formulas in \CLPSET are
built as in $\LCARD$. Hence, $\LCARD$ effectively extends \CLPSET by
introducing $\Size$ constraints and integer arithmetic. An $\LCARD$ formula not
including $\Size$ constraints nor integer constraints is a \CLPSET formula.
Hereafter, we will simply use the name \CLPSET to refer to the
constraint language offered by \setlog.
\qed
\end{remark}

\section{\label{satcard}$\SATCARD$: a constraint solving procedure for $\LCARD$}

A complete solver for \CLPSET is proposed in \cite{Dovier00}. In this section,
we show how that solver can be combined with Zarba's decision procedure
\cite{DBLP:conf/frocos/Zarba02}---hereafter simply called $\Zarba$---to support
cardinality constraints. The resulting constraint solving procedure, called
$\SATCARD$ (read `sat-card'), is a decision procedure for $\LCARD$ formulas.
Furthermore, it produces a finite representation of all possible solutions of
any satisifiable $\LCARD$ formula (see Section \ref{decproc}).

\subsection{The solver}

The overall organization of $\SATCARD$ is shown in Algorithm
\ref{glob}. Basically, $\SATCARD$ uses four routines:
\textsf{gen\_size\_leq}, $\mathsf{STEP_S}$, \textsf{remove\_neq} and
$\STEPCARD$. $\STEPCARD$, which is crucial for the integration of cardinality
constraints into \CLPSET, will be presented separately in Section
\ref{sizesolver}.

\begin{algorithm}
\begin{algorithmic}[0]
 \State $\Phi \gets \textsf{gen\_size\_leq}(\Phi)$;
 \Repeat
   \State $\Phi' \gets \Phi$;
   \Repeat
     \State $\Phi'' \gets \Phi$;
     \State $\Phi \gets \mathsf{STEP_S}(\Phi)$
     \hfill{\footnotesize[$\mathsf{STEP_S}$ returns $\false$ when $\Phi$ is unsat]}
   \Until{$\Phi = \Phi''$}
   \State $\Phi \gets \textsf{remove\_neq}(\Phi)$
 \Until{$\Phi = \Phi'$}
 \hfill{\footnotesize[end of main loop]}
 \State \textbf{let} $\Phi$ \textbf{be} $\Phi_1 \land \Phi_2$
 \hfill{\footnotesize [$\Phi_1$ contains $\Size$ relevant constraints, see Section \ref{integrating}]}
 \State $\Phi_1 \gets \STEPCARD(\Phi_1)$
 \hfill{\footnotesize[$\STEPCARD$ returns $\false$ when $\Phi_1$ is unsat]}
 \State\Return{$\Phi_1 \land \Phi_2$}
 \hfill{\footnotesize[returns $\false$ (unsat); or a disjunction of formulas representing all solutions]}
\end{algorithmic}
\caption{The solver $\SATCARD$. $\Phi$ is the input formula.} \label{glob}
\end{algorithm}

\textsf{gen\_size\_leq} simply adds integer constraints to the input formula
$\Phi$ to force the second argument of each $\Size$ constraint in $\Phi$ to be a
non-negative integer. $\mathsf{STEP_S}$ includes the constraint solving
procedure for the \CLPSET fragment as well as the constraint solving procedures
for cardinality constraints (see Section \ref{steps}). $\mathsf{STEP_S}$
applies specialized rewriting procedures to the current formula $\Phi$ and
returns either $\false$ or the modified formula. Each rewriting procedure
applies a few non-deterministic rewrite rules which reduce the syntactic
complexity of $\CARD$-constraints of one kind. \textsf{remove\_neq} deals with
the elimination of $\neq$ constraints involving set variables. Its purpose and
definition is made evident in \ref{inequalities}.

The execution of $\mathsf{STEP_S}$ and \textsf{remove\_neq} is iterated until a
fixpoint is reached, i.e., the formula is irreducible. These routines return
$\false$ whenever (at least) one of the involved procedures rewrites $\Phi$ to
$\false$. In this case, a fixpoint is immediately detected.

As we will show in Section \ref{decproc}, when all the non-deterministic
computations of $\SATCARD(\Phi)$ return $\false$, then we can conclude that
$\Phi$ is unsatisfiable; otherwise, we can conclude that $\Phi$ is satisfiable
and each solution of the formulas returned by $\SATCARD$ is a solution of
$\Phi$, and vice versa.

The rewrite rules used by $\SATCARD$ are defined as follows.

\begin{definition}[Rewrite rules]\label{d:rw_rules}
If $\pi$ is a symbol in $\Pi$ and $\phi$ is a $\CARD$-constraint based on
$\pi$, then a \emph{rewrite rule for $\pi$-constraints} is a rule of the form
$\phi \lfun \Phi_1 \lor \dots \lor \Phi_n$, where $\Phi_i$, $i \geq 1$, are
$\CARD$-formulas. Each $\Sigma_{\CARD}$-predicate matching $\phi$ is
non-deterministically rewritten to one of the $\Phi_i$ s. Variables appearing in
the right-hand side but not in the left-hand side are assumed to be fresh
variables, implicitly existentially quantified over each $\Phi_i$. \qed
\end{definition}

A \emph{rewriting procedure} for $\pi$-constraints consists of the collection
of all the rewrite rules for $\pi$-constraints. For each rewriting procedure,
$\mathsf{STEP_S}$ checks rules in the order they are listed
in the figures below. The first rule whose left-hand side matches the input
$\pi$-constraint is used to rewrite it.
Constraints that no rule rewrites are called \emph{irreducible}. Irreducible
constraints are part of the final answer of $\mathsf{STEP_S}$ (see Definition
\ref{def:solved}).

The following conventions are used throughout the rules. $\dot x$, for any name
$x$, is a shorthand for $x \in \Var$, i.e., $\dot x$ represents a variable. In
particular, variable names $\dot n$, $\dot n_i$, $\dot N$ and $\dot{N_i}$
denote fresh variables of sort $\sInt$ and $\sSet$, respectively. Moreover,
conjunctions occurring at the right-hand side of any given rule have higher
precedence than disjunctions.

\subsection{\label{steps}Set solving ($\mathsf{STEP_S}$)}

$\mathsf{STEP_S}$ can be divided into two collections of rewriting
procedures: those given as part of the \CLPSET system and those concerning
$\Size$ constraints.

The rewriting procedures of \CLPSET cover constraints based on $=$ when
arguments are either of sort $\sSet$ or $\sUr$, $\in$, $\Cup$, and $\disj$.
Figure \ref{f:clpset} lists some representative rewrite rules of \CLPSET
which, informally, work as follows:
\begin{itemize}
\item Rule \eqref{eq:ext} is the main rule of set unification. It states when two
non-empty, non-variable sets are equal by non-deterministically and recursively
computing four cases. These cases implement the \eqref{Ab} and \eqref{Cl}
axioms shown in Section \ref{semantics}. As an example, by applying rule
\eqref{eq:ext} to $\{1\} = \{1,1\}$ we get: ($1 = 1 \land \e = \{1\}) \lor (1 =
1 \land \{1\} = \{1\}) \lor (1 = 1 \land  \e = \{1,1\}) \lor (\e = \{1 \plus
\dot N\} \land  \{1 \plus \dot N\} = \{1\})$, which turns out to be true
(due to the second disjunct).

\item Rule \eqref{in:var} rewrites a set membership constraint into an equality
constraint. This means that a formula such as $x \in \dot{A} \land y \in
\dot{A}$ will eventually be transformed into $\{x \plus \dot{N_1}\} = \{y \plus
\dot{N_2}\}$ which will be processed by rule \eqref{eq:ext}.

\item Rule \eqref{notin:ext}
deals with not membership constraints. When
the r.h.s. of a $\notin$ constraint is an extensional set term, rule
\eqref{notin:ext} operates recursively to check that $x$ is not an element of
the set. Conversely, when the r.h.s. is a variable, $\notin$ constraint
are left unchanged (see Definition \ref{def:solved}).

\item Rule \eqref{un:ext1} is one of the main rules for $\Cup$ constraints.
Observe that this rule is based on set unification. It computes two cases: $x$
does not belong to $A$ and $x$ belongs to $A$ (in which case $A$ is of the form
$\{x \plus \dot{N_2}\}$ for some set $\dot{N_2}$).  In the latter case $x
\notin \dot{N_2}$ prevents Algorithm \ref{glob} from generating infinite terms
denoting the same set.

\item Finally, rule \eqref{disj:id} deals with a particular form of a disjointness
constraint.
\end{itemize}
The rest of the rewrite rules of \CLPSET can be found in \cite{Dovier00} and
online \cite{calculusBR}.

\begin{figure}
\hrule\vspace{3mm}
\begin{flalign}
\quad\quad & \{x  \plus{} A\} = \{y \plus B\} \lfun & \notag  \\
  & \qquad x = y \land A = B & \notag \\
  & \qquad \lor x = y \land \{x \plus A\} = B & \label{eq:ext} \\
  & \qquad \lor x = y \land A = \{y \plus B\} & \notag \\
  & \qquad \lor A = \{y \plus \dot N\} \land \{x \plus \dot N\} = B & \notag \\[2mm]
& x \in \dot{A} \lfun \dot{A} = \{x \plus \dot N\} & \label{in:var}  \\[2mm]
& x \notin \{y \plus A\} \lfun x \neq y \land x \notin A  & \label{notin:ext}  \\[2mm]
& \Cup(\{x \plus C\}, A, \dot{B}) \rightarrow & \notag \\
  & \qquad  \{x \plus C\} = \{x \plus \dot{N_1}\}
      \land x \notin \dot{N_1} \land \dot{B} = \{x \plus \dot N\} & \label{un:ext1} \\
  & \qquad \land (x \notin A \land \Cup(\dot{N_1}, A, \dot N) & \notag \\
  & \qquad {}\qquad\lor A = \{x \plus \dot{N_2}\}
                 \land x \notin \dot{N_2} \land \Cup(\dot{N_1}, \dot{N_2}, \dot
                 N)) & \notag \\[2mm]
& \dot{X} \disj \dot{X} \rightarrow \dot{X} = \e & \label{disj:id}
\end{flalign}
 \hrule
 \caption{\label{f:clpset}Some rewrite rules of \CLPSET}
\end{figure}

The rewrite rules concerning $\Size$ constraints implemented in
$\mathsf{STEP_S}$ are listed in Figure \ref{f:card1}. Rules
\eqref{size:empty}-\eqref{size:expr} are straightforward. Rule \eqref{size:ext}
computes the size of any extensional set by counting the elements that belong
to it while taking care of avoiding duplicates. This means that, for instance,
the first non-deterministic choice for a formula such as
$\Size(\{1,2,3,1,4\},m)$ will be:
\[
1 \notin \{2,3,1,4\}
  \land m = 1 + \dot{n} \land \Size(\{2,3,1,4\},\dot{n}) \land 0 \leq \dot{n}
\]
which will eventually lead to a failure due to the presence of  $1 \notin
\{2,3,1,4\}$ and rule \eqref{notin:ext}. This implies that $1$ will be counted
in its second occurrence. Besides, the second choice becomes
$\Size(\{2,3,1,4\},m)$ which is correct given that $\card{\{1,2,3,1,4\}} =
\card{\{2,3,1,4\}}$.

\begin{figure}
\hrule\vspace{3mm}
\begin{flalign}
 \quad\quad
 & \Size(\emptyset,m) \lfun m = 0 \label{size:empty} & \\[2mm]
 & \Size(A,0) \lfun A = \emptyset \label{size:zero} & \\[2mm]
 & \text{If $e$ is a compound arithmetic expression:} \notag & \\
 & \qquad \Size(A,e) \lfun \Size(A,\dot{n}) \land  \dot{n} = e \land 0 \leq \dot{n} \label{size:expr} & \\[2mm]
& \Size(\{x \plus A\},m) \lfun & \notag \\
   & \qquad x \notin A \land m = 1 + \dot n \land \Size(A,\dot n) \land 0 \leq \dot{n} & \label{size:ext} \\
   & \qquad \lor A = \{x \plus \dot N\} \land x \notin \dot N \land \Size(\dot N,m)
      & \notag
\end{flalign}
\hrule \caption{\label{f:card1}Rewrite rules for the $\Size$ constraint}
\end{figure}

Integer constraints, i.e., atomic constraints whose arguments are of sort
$\sInt$ (including those based on $=$ and $\neq$),
are simply dealt with as irreducible by $\mathsf{STEP_S}$; hence, they are
passed ahead to be checked by the routine $\STEPCARD$ after the main loop of
$\SATCARD$ terminates successfully.

\subsection{\label{irreducible}Irreducible constraints}

When no rewrite rule is applicable to the current $\CARD$-formula $\Phi$ and $\Phi$
is not $\false$, the main loop of $\SATCARD$ terminates returning $\Phi$ as its
result. This formula can be seen, without loss of generality, as $\Phi_\Set
\land \Phi_\Int$, where $\Phi_\Int$ contains all (and only) integer constraints
and $\Phi_\Set$ contains all other constraints occurring in $\Phi$.

The following definition precisely characterizes the form of atomic constraints
in $\Phi_\Set$.

\begin{definition}[Irreducible formula]\label{def:solved}
Let $\Phi$ be a $\CARD$-formula, $A$ and $A_i$ $\CARD$-terms of sort $\sSet$, $t$
and $\dot{X}$ $\CARD$-terms of sort $\langle \{\sSet,\sUr\} \rangle$, $x$ a
$\CARD$-term of any sort, and $c$ a variable or a constant integer number. A
$\CARD$-constraint $\phi$ occurring in $\Phi$ is \emph{irreducible} if it has
one of the following forms:
\begin{enumerate}[label=(\roman*), leftmargin=*, widest=viii]
\item \label{i:icfirst} $\dot{X} = t$, and neither $t$ nor $\Phi \setminus \{\phi\}$
contains $\dot{X}$;
\item $\dot{X} \neq t$, and $\dot{X}$ does not occur either in $t$ or
as an argument of any constraint $\pi(\dots)$, $\pi \in \{\Cup,\Size\}$, in
$\Phi$;
\item $x \notin \dot{A}$, and $\dot{A}$ does not occur in $x$;
\item $\Cup(\dot{A}_1,\dot{A}_2,\dot{A}_3)$, where $\dot{A}_1$ and $\dot{A}_2$
are distinct variables;
\item $\dot{A}_1 \disj \dot{A}_2$, where $\dot{A}_1$ and $\dot{A}_2$ are distinct variables;
\item $\Size(\dot{A}, c)$, $c \neq 0$.
\end{enumerate}
A $\CARD$-formula $\Phi$ is irreducible if it is $\true$ or if all of its
$\CARD$-constraints are irreducible. \qed
\end{definition}

$\Phi_\Set$, as returned by $\SATCARD$ once it finishes its main loop, is an
irreducible formula. This fact can be checked by inspecting the rewrite rules
presented in \cite{Dovier00} and those for the $\Size$ constraints given in
Figure \ref{f:card1}. This inspection is straightforward as there are no
rewrite rules dealing with irreducible constraints and all non-irreducible form
constraints are dealt with by some rule.

Putting $\Size$ constraints aside, $\Phi_\Set$ is basically the formula
returned by the \CLPSET solver. \citeN[Theorem 9.4]{Dovier00} show
that such formula is always satisfiable, unless the result is $false$.

It is important to observe that the atomic constraints occurring in $\Phi_\Set$
are indeed quite simple. In particular, all non-variable set terms occurring in the input formula have been removed,
except those occurring as right-hand sides of $=$ and $\neq$ constraints. Thus,
all (possibly complex) equalities and inequalities between set terms have been
solved. Furthermore, all arguments of $\Cup$ and $\disj$ constraints are
necessarily simple variables.

\section{\label{sizesolver}Cardinality solving ($\STEPCARD$)}

Due to the presence of $\Size$ and integer constraints, a non-$false$ formula
returned by $\mathsf{STEP_S}$ and $\mathsf{remove\_neq}$ is not always
satisfiable.

\begin{example}
Assuming all the arguments to be variables, the following formula cannot be
processed any further by $\mathsf{STEP_S}$ but is unsatisfiable:
\begin{equation*}
\Cup(A,B,C) \land \Size(A,m_a) \land \Size(B,m_b) \land \Size(C,m_c)
\land m_a + m_b < m_c
\end{equation*}
as it states that $\card{A} + \card{B} < \card{A \cup B}$.
\qed
\end{example}

Therefore, Algorithm \ref{glob} includes a new step, called $\STEPCARD$, whose
purpose is to check satisfiability of the formula returned at the end of the
main loop of $\SATCARD$.

Basically, $\STEPCARD$ encodes an adaptation of the $\Zarba$ algorithm to our
CLP system. In order to explain how we adapted $\Zarba$ we first introduce it
briefly; some technical details are omitted to simplify the presentation.

\subsection{\label{zarba}An algorithm for deciding set formulas with cardinality}

The language considered by Zarba--hereafter simply called $\LZa$---includes the
following function symbols: $\e$, $\cup$, $\cap$, $\setminus$, $+$, $-$ and
$\card{\cdot}$; the usual predicate symbols: $=$, $\in$, $<$, $>$; and
variables and integer constants as usual. All symbols have standard sorts and
semantics; in particular, sets are finite. The language also includes the
singleton set symbol $\{\cdot\}$ to form extensional sets. Note that although
$\LZa$ does not include an integer product symbol, it still allows the
representation of expressions of the form $c*x$, with either $c$ or $x$ a
constant. Formulas in $\LZa$ are built in the usual way.

$\Zarba$ is divided into four phases and takes as input a conjunction of
$\LZa$ literals. However, we will present the last two phases as a single one.
\begin{enumerate}
\item
\textsc{First phase}. The input formula, $\Psi$, is
transformed and divided into two subformulas, $\Psi'$ and $\Psi''$. $\Psi''$
contains only literals of the form $v = \card{x}$ where $v$ and $x$ are integer
and set variables, respectively. $\Psi'$ contains the integer constraints
present in $\Psi$ plus a transformation of the set constraints in $\Psi$. This
transformation guarantees that all set constraints are of the following forms:
$x = y$, $x \neq y$, $x = \{u\}$, $x = y \cup z$, $x = y \cap z$ and $x = y
\setminus z$, where $x$, $y$ and $z$ are set variables and $u$ is a
ur-variable.

\begin{example}
A constraint such as $y \in x$ is transformed into $x = \{y\} \cup x$ and  then
into $w = \{y\} \land x = w \cup x$, where $w$ is a new variable.

A constraint such as $\{u\} \cup x = h \cap w$ is transformed into  $v = \{u\}
\land v \cup x = h \cap w$ and then into $v = \{u\} \land t = v \cup x \land t
= h \cap w$, where $v$ and $t$ are new variables.

A constraint such as $\card{x} + m < k$ is transformed into  $v = \card{x}
\land v + m < k$, where $v$ is a new variable. In this way, $v = \card{x}$
becomes part of $\Psi''$. \qed
\end{example}

\item
\textsc{Second phase}. $\Psi'$ is divided into three subformulas: $\Psi_\Ur$,
containing literals of the form $x = \{u\}$, where $u$ is a ur-element;
$\Psi_\Int$, containing the integer literals; and $\Psi_\Set$, containing the
set literals. So now the input formula has been transformed and divided into
four subformulas: $\Psi_\Ur$, $\Psi_\Int$, $\Psi_\Set$ and $\Psi''$. In the
next phase, $\Psi
\defs \Psi_\Ur \land \Psi_\Int \land \Psi_\Set \land \Psi''$.
\item
\textsc{Third phase}. This phase consists in executing the following  three
steps for each \emph{arrangement} of $\Psi$. Whenever there are no more
arrangements the input formula is unsatisfiable.

An arrangement of $\Psi$ is a tuple $\langle R, \Pi, at\rangle$ where: $R
\subseteq \Var_\Ur^\Psi \times \Var_\Ur^\Psi$ is an equivalence relation where
$\Var_\Ur^\Psi$ is the collection of ur-variables in $\Psi_\Ur$; $\Pi$ is a
finite collection of non-$\false$ Boolean functions $\pi: \Var_\Set^\Psi \fun
\{0,1\}$ where $\Var_\Set^\Psi$ is the collection of set variables in
$\Psi_\Set \land \Psi''$; and $at: \Var_\Ur^\Psi \fun \Pi$. $\pi$ is a
non-$\false$ Boolean function if $1 \in \ran \pi$.

From now on $\rho = \langle R, \Pi, at\rangle$ denotes the current arrangement.
   \begin{enumerate}[ref=\alph*]
   \item\label{i:checkarr}
   In this step the algorithm checks whether or not $\rho$ verifies seven
   conditions. If $\rho$ does not verify these conditions the next arrangement
   is chosen; if it does the next step is executed. We show only the
   conditions that are used in our implementation.
      \begin{enumerate}
      \renewcommand{\labelenumii}{\roman{enumii}}
      \item\label{i:un}
      If $x = y \cup z$ is in $\Psi_\Set$ then $\pi(x) = 1$ if and only if
      $\pi(y) = 1$ or $\pi(z) = 1$, for each $\pi \in \Pi$.
      \item\label{i:disj}
      If $\e = y \cap z$ is in $\Psi_\Set$ then $\pi(y) = 0$ or $\pi(z) = 0$.
      \end{enumerate}
   The remaining conditions are not used because $\Zarba$ is called after
   $\mathsf{STEP_S}$; see Section \ref{integrating} for more details.
   \item\label{i:checkint}
   In this step the algorithm checks whether or not $\Psi_\Int \land
   \mathit{res}_\Int(\rho)$ is satisfiable, where:
    \begin{equation}\label{eq:res}
    \begin{split}
    \mathit{res}_\Int&(\rho) \defs \\
      & \bigwedge_{\pi \in \Pi} 0 < v_\pi
      \bigwedge_{\pi \in \ran at} v_\pi = 1
      \bigwedge_{v = \card{x} \in \Psi''}
        v = \sum_{\pi \in \Pi} \pi(x) * v_\pi
    \end{split}
    \end{equation}
   If $\Psi_\Int \land \mathit{res}_\Int(\rho)$ is unsatisfiable the
   next arrangement is chosen and step \eqref{i:checkarr} is executed.
   \item
   In this last step the algorithm checks whether or not there are enough
   ur-elements as to satisfy $\Psi_\Ur$ considering the equivalence relation $R$
   of $\rho$ and the minimum of $\sum_{\pi \in \Pi} v_\pi$ subject to $\Psi_\Int
   \land \mathit{res}_\Int(\rho)$. If this is satisfiable, the input formula is
   satisfiable; if not, the next arrangement is chosen and step \eqref{i:checkarr}
   is executed.
   \end{enumerate}

   Informally, in this phase the algorithm assigns a positive cardinality ($v_\pi$)
   to each non-empty Venn region involved in the formula and tries, one after the
   other, all possible combinations of these assignments---each combination is
   encoded in each arrangement. With each combination it builds formula
   \eqref{eq:res} and checks whether the cardinality constrains are
   satisfiable or not.
\end{enumerate}

\subsection{Integrating $\Zarba$ into $\SATCARD$}\label{integrating}  

The repeated execution of $\mathsf{STEP_S}$ and \textsf{remove\_neq} in
$\SATCARD$ implements up to the second phase of $\Zarba$. The third phase
of $\Zarba$ is implemented by $\STEPCARD$. Formulas returned at the end of the
main loop of $\SATCARD$ (i.e., $\CARD$-formulas in irreducible form) can be
easily transformed into the formulas obtained after executing the second phase
of $\Zarba$. A detailed definition of a mapping of such formulas into the
corresponding $\LZa$ formulas is given in \ref{mapping}. Hereafter, we provide
an intuitive description of which formulas are passed to
$\STEPCARD$.

Let $\Phi \defs \Phi_1 \land \Phi_2$ be the formula in irreducible form right
after the main loop of Algorithm \ref{glob}, where $\Phi_1$ contains all
integer constraints and all of the $\Cup$, $\disj$ and $\Size$ constraints, and
$\Phi_2$ is the rest of $\Phi$ (i.e., $\notin$ constraints, and $=$ and $\neq$
constraints not involving integer terms). Hence, $\STEPCARD$ is called on
$\Phi_1$ as follows:
\begin{itemize}
\item
All integer constraints are passed basically unaltered to
$\STEPCARD$.
\item
$\CARD$-constraints of the form $\Cup(A,B,C)$, $A \disj B$, $\Size(A,m)$, where
$A, B, C$ are variables and $m$ is either a variable or an integer constant,
are mapped to literals of the form $C = A \cup B$, $A \cap B = \e$, $\card{A} =
m$, respectively, in $\LZa$.
\end{itemize}
On the other hand, constraints in $\Phi_2$ are
not passed to $\STEPCARD$:
\begin{itemize}
\item equality constraints are ignored because these
variables do not appear in the rest of $\Phi$.
\item
$\neq$ constraints not involving integer terms and $\notin$ constraints are
ignored because they do not affect the cardinality of the set variables
involved in the formula. Indeed, in $\LCARD$ we assume that the universe of
objects which can be used as set elements is infinite---as it includes integer
numbers and (nested) sets. Hence, constraints of the form $X \neq t$ and $t
\notin X$ (with $X$ variable and $t$ any term) do not forbid any value of the
cardinality of $X$. For instance, if $\Phi$ contains $1 \notin S \land 2 \notin
S \land ... \land 20 \notin S \land \Size(S,m)$, with $S$ variable, then we can
find anyway $m$ constants different from $1,\dots,20$ to fill the set $S$.
\end{itemize}

Note that non-variable set terms occur only in those
constraints of $\Phi_2$ that are not passed to $\STEPCARD$. Thus, the
translation function $\map$ shown in \ref{mapping}, which only deals with
variables, is indeed capable of translating any $\LCARD$ formula that is passed
to it.

$\STEPCARD$ implements the first two steps of the third phase by casting step
\eqref{i:checkarr} in terms of a Boolean satisfiability problem and step
\eqref{i:checkint} in terms of an integer linear programming (ILP) problem
\cite{10.5555/1611284}. All the solutions returned by solving the Boolean
formula are collected in a set $S$ and then all possible
arrangements are the elements of $2^{S}$.
A description of a concrete implementation of these two steps is given in the
next subsection.

The last step of the third phase is not implemented again because of
the assumption about the infinity of the universe of objects which can be used
as set elements in $\LCARD$.

It is worth noting that, in the integrated system, unsatisfiability
caused by set constraints, excluding $\Size$, can be caught directly by
$\mathsf{STEP_S}$ and \textsf{remove\_neq}, without executing
$\STEPCARD$.

\begin{example}
Consider the following formula:
\begin{equation*}
\Cup(A,B,C) \land A \disj C \land A \neq \e \land \Size(C,k) \land k < 2.
\end{equation*}
where $A$, $B$, $C$, and $k$ are variables. The subformula $\Cup(A,B,C) \land A
\disj C \land A \neq \e$ is not in irreducible form and it is further processed
first by \textsf{remove\_neq} and then by $\mathsf{STEP_S}$, that finally
rewrites it to \false. That is, the input formula is found to be unsatisfiable
disregarding the cardinality and integer constraints occurring in it. \qed
\end{example}

On the other hand, the presence of $\STEPCARD$ in $\SATCARD$ allows us to solve
linear integer constraints even if the given formula does not contain any
$\Size$ constraint. For example, a formula such as $x > y \land x < y +
1$ is found to be false by exploiting the integer constraint solver included
in $\STEPCARD$.

\subsection{A concrete implementation of $\STEPCARD$}\label{concrete}

In this section we briefly outline a concrete Prolog implementation of
$\STEPCARD$. This implementation is obtained by integrating into the
$\STEPCARD$ procedure described above a Prolog Boolean SAT solver, namely the
very concise solver developed by
\citeN{DBLP:journals/tcs/HoweK12}, and the implementation of the CLP(Q) system
of SWI-Prolog \cite{holzbaur1995ofai}.

CLP(Q) implements a solver for linear equations, a Simplex algorithm to decide
linear inequalities and a branch and bound method to provide a decision
algorithm for ILP. This library provides
$\mathtt{bb\_inf}(Vars,Expr,Min,Vert)$, which finds the vertex ($Vert$) of the
minimum ($Min$) of the expression $Expr$ subjected to the integer constraints
present in the constraint store and assuming all the variables in $Vars$ take
integers values. In its way to find the minimum value, $\mathtt{bb\_inf}$ first
determines whether or not the constraints are satisfiable (in $\num$).
$\mathtt{bb\_inf}$ is complete provided all integer constraints are linear.
With respect to the completeness of $\mathtt{bb\_inf}$, observe that: \emph{a)}
$\LCARD$ restricts integer constraints to be linear (Definition
\ref{LCARD-terms}); and \emph{b)} the integer constraints generated by any rule
for $\Size$ are linear.

Consider a formula $\Phi$ received by $\STEPCARD$. Now consider the subformula
of $\Phi$ that is a conjunction of constraints of the following forms:
$\Cup(A,B,C)$ and $A \disj B$, with $A, B$ and $C$ variables. As $\Zarba$
must find all the non-$\false$ Boolean functions $\pi:\Var_\Set^\Phi \fun
\{0,1\}$ verifying some Boolean conditions (see Section \ref{sizesolver} for
some examples and \cite[conditions (C1)-(C7) in
3.4]{DBLP:conf/frocos/Zarba02}), we encode the conjunction of these constraints
as a Boolean formula as follows:
\begin{itemize}
  \item $\Cup(A,B,C) \lfun (\lnot C \lor B \lor A)
  \land (\lnot A \lor C) \land (C \lor \lnot A)$,
  due to condition \ref{i:un}.
  \item $A \disj B \lfun \lnot A \lor \lnot B$,
  due to condition \ref{i:disj}.
  \end{itemize}

Next, we call Howe and King's SAT solver on the resulting Boolean formula and
collect in a set $S$ all the Boolean solutions where at least one variable is
bound to $\true$. Hence, $S$ contains all possible non-$\false$ Boolean
functions $\pi:\Var_\Set^\Phi \fun \{0,1\}$ satisfying $\Zarba$'s  conditions \ref{i:un}
and \ref{i:disj}.

If $\{\pi_1,\dots,\pi_n\}$ verifies the above condition, then we use it to
execute the second step of the third phase. Then we build formula
\eqref{eq:res} as a conjunction of CLP(Q) constraints, which is easy to
implement. All the integer constraints present in $\Phi$ and all those in
\eqref{eq:res} are passed in to the CLP(Q) constraint store. Finally, we call
CLP(Q)'s \verb+bb_inf/4+ predicate\footnote{\texttt{bb\_inf/4}:
\url{https://www.swi-prolog.org/pldoc/doc_for?object=bb_inf/4}} as follows:
\begin{equation}\label{eq:bbinf}
\mathtt{bb\_inf}(\Var_\Int,\sum_{i=1}^k m_i,\_,Vertex)
\end{equation}
where $m_1,\dots,m_k$ are the second arguments of the $\Size$ constraints in
$\Phi$. That is, we ask CLP(Q) to check the satisfiability of its constraint
store assuming that all the variables there are integers, and if so, to compute
the vertex ($Vertex$) of the minimum of the sum of the cardinalities of the
sets in $\Phi$. If this call does not fail we know $\Phi$ is satisfiable and
$\STEPCARD$ terminates; if not, we pick the next subset of $S$. If $\STEPCARD$
fails for all subsets of $S$ it returns $\false$.

\section{\label{decproc}$\SATCARD$ is a decision procedure for $\LCARD$}

In this section we analyze the soundness, completeness and termination
properties of $\SATCARD$.

The following theorem ensures that, after termination, the
rewriting process implemented by $\SATCARD$ preserves the set of solutions of
the input formula.

\begin{theorem}[Equisatisfiability]\label{equisatisfiable}
Let $\Phi$ be a $\CARD$-formula and $\Phi^1, \Phi^2,\dots,\Phi^n$ be the
collection of $\CARD$-formulas returned by $\SATCARD(\Phi)$. Then $\Phi^1 \lor
\Phi^2 \lor \dots \lor \Phi^n$ is equisatisfiable to $\Phi$, that is, every
possible solution\footnote{More precisely, each solution of $\Phi$ expanded to
the variables occurring in $\Phi^i$ but not in $\Phi$, so as to account for the
possible fresh variables introduced into $\Phi^i$.} of $\Phi$ is a solution of
one of the $\Phi^i$s and, vice versa, every solution of one of these formulas is
a solution for $\Phi$.
\end{theorem}

\begin{proof}
According to Definition \ref{irreducible}, each formula
$\Phi_i$ returned at the end of $\SATCARD$'s main loop is of the form
$\Phi^i_\Set \land \Phi^i_\Int$, where $\Phi^i_\Set$ is a $\CARD$-formula
in irreducible form and $\Phi^i_\Int$ contains all integer constraints
encountered during the processing of the input formula.
As concerns constraints in $\Phi_i^\Set$, the proof is based on showing that
for each rewrite rule the set of solutions of left and right-hand sides is the
same. For those rules dealing with constraints different from $\Size$ the
proofs can be found in \cite{Dovier00}. The proofs of equisatisfiability for
the rules for $\Size$ can be found in \ref{proofs}.
As concerns $\Phi^i_\Int$, no rewriting is actually performed on the
constraints occurring in it. Thus the set of solutions is trivially preserved.
Considering also the last step of $\SATCARD$, i.e., calling $\STEPCARD$, we
observe that this step is just a check which either returns $\false$ or has no
influence on its input formula.
\end{proof}



\begin{theorem}[Satisfiability of the output formula]\label{satisf}
Any $\CARD$-formula different from $\false$ returned by $\SATCARD$ is
satisfiable w.r.t. the underlying interpretation structure $\iS$.
\end{theorem}

\begin{proof}
Basically, the proof of this theorem relies on the fact that $\STEPCARD$
implements $\Zarba$. Let $\Phi$ be the input formula and $\Phi'$ its
irreducible form right before $\STEPCARD$. Consider that $\Phi'$ is divided as
$\Phi'_1 \land \Phi'_2$ where $\Phi'_1$ contains the integer constraints and
the $\Cup$, $\disj$ and $\Size$ constraints; and $\Phi'_2$ all the other
constraints. Then, $\Phi'_1$ can be easily mapped to formulas which are
accepted by $\Zarba$ (see \ref{mapping}). As observed in Section
\ref{integrating}, $\Phi'_2$ is not passed to $\Zarba$
because is irrelevant as regards the satisfiability of $\Phi'_1$.  Then, the
satisfiability of $\Phi$ depends only on the satisfibility of $\Phi'_1$. Hence,
if $\STEPCARD$ decides that $\Phi'_1$ is satisfiable, we can conclude that
$\Phi$ is satisfiable. In this case $\SATCARD$ returns $\Phi$.
\end{proof}

Thanks to Theorems \ref{equisatisfiable} and \ref{satisf} we can conclude that,
given a $\CARD$-formula $\Phi$, then $\Phi$ is
satisfiable with respect to the intended interpretation structure $\iS$ if and
only if there is a non-deterministic choice in $\SATCARD(\Phi)$ that returns a
$\CARD$-formula different from $\false$. Conversely, if all the
non-deterministic computations of $\SATCARD(\Phi)$ terminate with $\false$,
then $\Phi$ is surely unsatisfiable.


The following is an example of a formula that $\SATCARD$ is able to detect to
be unsatisfiable.
\begin{example}
The formula
\begin{equation*}
\Cup(A,B,C) \land \Size(A,m_1) \land \Size(B,m_2) \land \Size(B,m_3) \land m_3
> m_1 + m_2
\end{equation*}
where all arguments are variables, is rewritten by $\SATCARD$ to $\false$;
hence, the formula is unsatisfiable. \qed
\end{example}

Note that many of the rewriting procedures given in the previous section will
stop even when returning relatively complex formulas.

\begin{example}\label{ex:solution}
Assuming all the arguments are variables, the formula:
\begin{equation*}
\Cup(A,B,C) \land \Size(A,m_1) \land \Size(B,m_2) \land \Size(B,m_3) \land m_3
\leq m_1 + m_2
\end{equation*}
is returned unchanged by $\SATCARD$ because there is no rewrite rule for
constraints such as $\Cup(A,B,C)$ and $\Size(A,m)$ when all arguments of sort
$\sSet$ are variables. Actually, this formula is proved to be satisfiable by applying
$\STEPCARD$. \qed
\end{example}

Finally, we can state the termination property for $\SATCARD$.

\begin{theorem}[Termination] \label{termination-glob}
The $\SATCARD$ procedure can be implemented as to ensure termination
for every input $\LCARD$ formula.
\end{theorem}

\begin{proof}
Termination of the $\SATCARD$ is a consequence of the termination
proved in Theorem 10.10 in \cite{Dovier00} and Zarba's
algorithm \cite[Theorem 3]{DBLP:conf/frocos/Zarba02}. The only new observations
to be done concern the treatment of $\Size$ constraints. Looking at the rewrite
rules for this kind of constraints shown in Figure \ref{f:card1}, we can
observe that: they generate equality and inequality constraints (in fact,
$\not\in$ constraints are rewritten to $\neq$ constraints), which in turn do
not generate any new $\Size$ constraint; besides, they generate new $\Size$
constraints which, however, are in irreducible form, since their first argument
is a (fresh) variable. Therefore, the processing of $\Size$ constraints cannot
trigger any infinite loop.
\end{proof}

\section{\label{minimal}Minimal solutions}

The formulas $\Phi_1,\dots,\Phi_n, n \ge 1,$ returned by $\SATCARD$ represent
all the concrete (or ground) \emph{solutions} of the input formula $\Phi$. If
these formulas do not contain any $\Size$ or integer constraints, then it is
quite easy to get concrete solutions from them. Indeed, a successful assignment
of values to variables (i.e., a concrete solution) for such formulas is
obtained by substituting each set variable occurring in them by the empty set,
with the exception of the variables $X$ in atoms of the form $X = t$.

Unfortunately, when it comes to the $size$ and integer constraints, providing
concrete solutions for certain $\LCARD$-formulas may be difficult.

\begin{example}\label{ex:size_slvfrm}
If $\SATCARD$ is called on the following formula:
\begin{equation*}
\Size(A,m) \land 1 \leq m \land B \subseteq A
\land \Size(B,n) \land 5 \leq n
\end{equation*}
it will return the same formula meaning that it is satisfiable. However, a
solution is not evident. \qed
\end{example}

For some applications such as model-based testing \cite{CristiaRossiSEFM13}
determining the satisfiability of a formula is not enough. More explicit
solutions are needed. For this reason we provide a way in which $\SATCARD$
returns formulas for which finding a solution is always easy. We call such a
solution \emph{minimal} because no cardinality of a set assigned to a variable
appearing in a $size$ constraint can be lowered without making the formula
false. However, in this case we cannot get a finite representation of the set
of all possible solutions.

Let $\Phi$ be a satisfiable input formula and let $\Phi'$ the corresponding
formula right before $\STEPCARD$ is called. Let $\Size(A_1,m_1), \dots,
\Size(A_k,m_k)$ be all the $\Size$ constraints in $\Phi'$. If $\SATCARD$ is
required to compute the minimal solution, once Algorithm \ref{glob} finishes,
it is called again with the following formula:
\begin{equation}\label{eq:minimal}
  \Phi' \land \bigwedge_{i=1}^{k} m_i = V_i
\end{equation}
where $\langle V_1,\dots,V_k \rangle$ is the $Vertex$ computed in
\eqref{eq:bbinf}. In this way all sets $A_i$ of the $\Size$ constraints in
$\Phi'$ are bound to bounded sets of least possible cardinality so as to
satisfy $\Phi$. Note that, necessarily, $0 \leq V_i$, for $i \in [1,k]$.

Besides, when $\SATCARD$ runs in this mode it will not call $\STEPCARD$ to
solve \eqref{eq:minimal}. In fact, $\bigwedge_{i=1}^{k} m_i = V_i$ turns all
$\Size$ constraints in $\Phi'$ into atoms of the form $\Size(\dot{A},c)$ with
$c$ a constant. Then, the following rewrite rule is activated:
\begin{equation}
\Size(\dot A,c) \text{, $c$ an integer constant} \lfun
  \dot A = \{\dot{n}_1,\dots,\dot{n}_c\}
  \land ad(\dot{n}_1,\dots,\dot{n}_c) \label{size:const3}
\end{equation}
where $ad(y_1,\dots,y_c)$ is a shorthand for $\bigwedge_{i=1}^{c-1}
\bigwedge_{j=i+1}^c y_i \neq y_j$.

\begin{example}\label{ex:minsol}
If $\SATCARD$ is called on the formula of Example \ref{ex:size_slvfrm} but
requiring that all minimal solutions be computed, then the formula returned at
the end of the computation is:
\begin{flalign*}
  \quad\quad & A = \{n_5,n_4,n_3,n_2,n_1\}, m = 5,
   B = \{n_5,n_4,n_3,n_2,n_1\}, n = 5,
   ad(n_5,n_4,n_3,n_2,n_1) &
\end{flalign*}
This formula is a finite representation of a subset of the
possible solutions for the input formula
from which it is trivial to get concrete solutions. \qed
\end{example}

\section{\label{impl} The Implementation and its Empirical Evaluation}

$\SATCARD$ is implemented by extending the solver provided by the publicly
available tool \setlog \cite{setlog}. \setlog is a Prolog program that can be
used as a constraint solver, as  a satisfiability solver and as a constraint
logic programming language. It also provides some programming facilities not
described in this paper. In this section we describe and empirically evaluate
this implementation.

The main syntactic differences between the abstract syntax used in previous
sections and the concrete syntax used in \setlog are made evident by the
following examples.

\begin{example}\label{ex:setlogformulas}
The formulas of Example \ref{ex:formulas} are written in \setlog as follows:
\begin{verbatim}
   a in A & a nin B & un(A,B,C) & C = {X / D}.

   un(A,B,C) & N + K > 5 & size(C,N) & B neq {}.
\end{verbatim}
where names beginning with a capital letter represent variables, and all others
represent constants and function symbols. This is why we renamed some variables
w.r.t. the formulas in Example \ref{ex:formulas}. Note that \verb+{_/_}+ is the
concrete syntax for the set term $\{\_\plus\_\}$.

If \setlog is asked to solve the second formula it returns the following:
\begin{verbatim}
   B = {_N3/_N2}, C = {_N3/_N1}
   Constraint: un(A,_N2,_N1), N + K > 5, _N3 nin _N1,
               size(_N1,_N4), _N4 >= 0, N >= 1, _N4 is N - 1
\end{verbatim}
as the first solution (more can be obtained interactively). That is, \setlog
binds values to $B$ and $C$ and gives a list of constraints in irreducible form
(which is guaranteed to be satisfiable). Any concrete solution must bind values
to the remaining variables in such a way as to verify the constraints.
Variables beginning with the underscore symbol (\verb+_+) represent new
variables.
\qed
\end{example}

The implementation in \setlog of $\mathsf{STEP_S}$ consists in  adding to
\setlog the rewrite rules of Figure \ref{f:card1}. Due to the design of
\setlog, adding new constraints and their rewrite rules is easy, and it does
not deserve to be further commented here. On the other hand, the implementation
in \setlog of $\STEPCARD$ is basically that described in Section
\ref{concrete}.

Observe
that the fact that \setlog is based on set unification automatically provides
cardinality over sets of sets---nested at any level. For instance, running
\verb+size({{X},{Y}},N)+ produces two solutions:
\begin{verbatim}
   N = 2, X neq Y;
   Y = X, N = 1
\end{verbatim}

Concerning formulas with $\Size$ constraints, by default \setlog decides their
satisfiability as described in Section \ref{sizesolver}. That is, if the
formula of Example \ref{ex:size_slvfrm} is executed, \setlog will find it
satisfiable and will return it unchanged. If users want more concrete
solutions, as described in Section \ref{minimal}, they must execute command
\verb+fix_size+ to activate the algorithm that computes minimal solutions.
In this case, after solving the formula of Example \ref{ex:size_slvfrm},
\setlog would return exactly the solution shown in Example \ref{ex:minsol}. As
another example, when solving the second formula of Example \ref{ex:setlogformulas}
 in \texttt{fix\_size} mode, \setlog will return (as its first
solution):
\begin{verbatim}
   A = {}, B = {_N1}, C = {_N1}, N = 1
   Constraint: 1 + K > 5
\end{verbatim}
which is indeed a more concrete solution for the given formula.

\subsection{\label{formver}Applications to formal verification}

We now present a simple example showing how \setlog can be used as a
verification tool of problems involving cardinality constraints. In doing so we
will show how our approach differs from other tools that can deal with similar
problems---see Section \ref{related} for a detailed account. More than 250
real-world examples have been used in the empirical evaluation presented in
Section \ref{empirical} and another example is developed in
\ref{app:datacontainer}. The example is taken from \citeN{Kuncak2006}. Figure
\ref{f:insert} shows the $\mathsf{insert}$ procedure which inserts an element
$e$ into the set $content$. Besides, the procedure maintains the cardinality of
$content$ in variable $size$. In this context an \textsf{element} is an object
represented as a set of cardinality one. The procedure is annotated with its
preconditions (i.e., \textsc{requires}), its postconditions (i.e.,
\textsc{ensures}) and the invariant it preserves (i.e., \textsc{maintains}).
Kuncak then proposes a verification condition for the \textsf{insert}
procedure.

\begin{figure}
\raggedright
\textsf{var} $content$:\textsf{set}; $size$:\textsf{integer}; \\[1mm]
\textsf{procedure insert}($e$:\textsf{element}) \{
  \hfill {\footnotesize
  [\textsc{requires:}
    $\card{e} = 1 \land \card{e \cap content} = 0$]}\\
\hspace{3mm}$content := content \cup e$;
  \hfill {\footnotesize
  [\textsc{maintains:} $size = \card{content}$] \\
\hspace{3mm}$size := size + 1$}; \\
\} \hfill {\footnotesize
   [\textsc{ensures:} $size' > 0$]}

\caption{\label{f:insert}Procedure $\mathsf{insert}$ inserts $e$ into set
$contents$ and updates its cardinality in $size$}
\end{figure}

The \setlog representation of \textsf{insert} is the following:
\begin{Verbatim}[commandchars=\\\$\$]
sl_insert(Content,Size,E,Content_,Size_) :-
  un(Content,E,Content_) & \hfill\small\rm [\(content := content \cup e\)]
  Size_ is Size + 1. \hfill\small\rm [\(size := size + 1\)]
\end{Verbatim}
where \verb+Content+ and \verb+Size+ are the initial values and
 \verb+Content_+ and \verb+Size_+ the final ones. In this way,
 \verb+sl_insert+ becomes a \setlog \emph{program} and thus it
can be  executed as any other program and can be part of a larger
Prolog+\setlog program. For example the query:
\begin{verbatim}
   sl_insert({},0,{hellow},C1,S1).
\end{verbatim}
returns:
\begin{verbatim}
   C1 = {hellow}, S1 = 1
\end{verbatim}
and the following one:
\begin{verbatim}
   sl_insert({},0,{hellow},C1,S1) & sl_insert(C1,S1,{world},C2,S2).
\end{verbatim}
returns:
\begin{verbatim}
   C1 = {hellow}, S1 = 1, C2 = {hellow,world}, S2 = 2
\end{verbatim}


Furthermore, \verb+sl_insert+ is also a \emph{formula}. Indeed, we can
discharge the verification condition indicated by Kuncak \emph{using the same
representation} of \textsf{insert} by simply encoding the negation of the
verification condition as a \setlog query:
\begin{Verbatim}[commandchars=\\\$\$]
   size(E,1) & inters(E,Content,M1) & size(M1,0) & \hfill\small \rm[precondition]
   size(Content,Size) & \hfill\small\rm[invariant@before state]
   sl_insert(Content,Size,E,Content_,Size_) & \hfill\small \rm[\sf insert\rm is executed]
   (Size_ =< 0 \hfill\small \rm[negation of postcondition@after state]
    or
    size(Content_,M2) & M2 neq Size_ \hfill\small \rm[negation of invariant@after state]
   ).
\end{Verbatim}
If the answer is \verb+no+ it means the query is unsatisfiable for all values
of the variables, and so the verification condition is a theorem. \setlog runs
this query in 0.016 seconds.

As the example shows, the \setlog representation of \textsf{insert} is both a
formula (or executable specification) \emph{and} a program (or prototype,
because of its lack of efficiency). Or put it in another way, \setlog is the
very same tool that executes \textsf{insert} and \emph{automatically} proves
its correctness. We think this is a rare characteristic in verification tools
dealing with cardinality constraints. \setlog has been used in the same fashion
on real-world problems (\citeANP{DBLP:journals/jar/CristiaR21}
\citeyearNP{DBLP:journals/jar/CristiaR21,Cristia2021}).

\subsection{Improvements}

In this section we present some improvements recently made to
\setlog to render it a more usable tool.

\paragraph{Derived constraints.}
As shown in Section \ref{expressiveness}, many set operators in
\setlog are defined as derived constraints, i.e., as $\CARD$-formulas built out
of the primitive constraints that $\LCARD$ offers. For example, the predicate
$\Cap(A,B,C)$, which is true when $C$ is the intersection between sets $A$ and
$B$, can be defined as a derived constraint as follows:
\[
\Cap(A,B,C) \defs
 \Cup(C,N_1,A) \land \Cup(C,N_2,B) \land N_1 \disj N_2
\]

This approach is good from a theoretical perspective because it keeps the
language, proofs and implementation to a minimum. However, it pays the price of
reduced efficiency which, in the end, makes the tool less interesting from a
practical perspective. Therefore, we move some key set constraints from derived
constraints to \emph{built-in} constraints by defining and implementing
possibly recursive rewriting procedures for them. Specifically, we select
$\Cap$, $\subseteq$, and $\Diff$ (for set difference) to be implemented as
built-in constraints. The main new rewrite rules for these constraints can be
found in an on-line document \cite{calculusBR}).

\paragraph{Inference rules.}

In order to further improve the efficiency of our solver we introduce
special rewrite rules---hereafter simply called \emph{inference rules}---that
allow new $\Size$ and integer constraints to be inferred from the irreducible
constraints. The presence of these additional constraints will allow the solver
to detect more efficiently certain classes of unsatisfiable formulas.

Some significant inference rules are shown in Figure \ref{fig:infrule}.

\begin{figure}
\hrule\vspace{3mm}
 \raggedright
 \quad\quad If $X$ is any of $\dot{A}_i$; $m$ is any of $\dot{m}_i$; $\dot{m}_i$ is the cardinality of
$\dot{A}_i$; then:
\begin{flalign}
 \quad\quad \quad\quad & \Cup(\dot{A}_1,\dot{A}_2,\dot{A}_3) \land \Size(X,\dot{m}) \lfun
\Cup(\dot{A}_1,\dot{A}_2,\dot{A}_3) \land \dot{m}_3 \leq
\dot{m}_1 + \dot{m}_2 \bigwedge_{i=1,2,3}\Size(\dot{A}_i,\dot{m}_i) & \label{rule:infer_un}  \\[2mm]
 \quad\quad \quad\quad & \Cap(\dot{A}_1,\dot{A}_2,\dot{A}_3) \land \Size(X,\dot{m}) \lfun
\Cap(\dot{A}_1,\dot{A}_2,\dot{A}_3) \bigwedge_{i=1,2,3}\Size(\dot{A}_i,\dot{m}_i) \bigwedge_{i=1,2} \dot{m}_3 \leq \dot{m}_i & \label{rule:infer_inters}
\end{flalign}
\hrule
 \caption{Rule scheme for $\Size$ inference rules}
\label{fig:infrule}
\end{figure}

\begin{example}
Proving a formula such as $ B = A_1 \cup \dots \cup A_{20} \land
\sum_{i=1}^{20} \card{A_i} < \card{B} $ which can be easily written in \setlog
by using $\Cup$, $\Size$, $=$ and $<$ constraints, would cause an exponential
explosion in $\STEPCARD$. Instead, by implementing the first inference rule
shown in Figure \ref{fig:infrule} the unsatisfiability is found in a few
milliseconds. In fact, the introduction of this rule eliminates the exponential
explosion for this class of formulas. \qed
\end{example}

Hence, we extend Algorithm \ref{glob} by properly adding new calls to the
inference rules inside $\STEPCARD$, just before starting the third phase of
$\Zarba$. If $\Phi_1$ is the formula received by $\STEPCARD$ and ${\Phi_1}'$
the one obtained from $\Phi_1$ after applying the inference rules, then CLP(Q)
is called on the integer subformula of ${\Phi_1}'$. If CLP(Q) fails, then the
whole computation fails and the input formula is unsatisfiable; if not, the
third phase of $\Zarba$ is started with $\Phi_1$.

\subsection{\label{empirical}Empirical evaluation}
In this section we present the results of the empirical evaluation we conducted
in order to evaluate how well the implementation of $\SATCARD$ in \setlog
performs in practice. In previous papers, we have evaluated other aspects of
\setlog such as its efficiency in producing model-based test cases
\cite{CristiaRossiSEFM13}; how well it deals with relational constraints
\cite{DBLP:journals/jar/CristiaR20} and restricted intensional sets
(\citeANP{DBLP:conf/cade/CristiaR17}
\citeyearNP{DBLP:conf/cade/CristiaR17,DBLP:journals/jar/CristiaR21a}); and we
have applied it to industrial-strength case studies such as the Bell-LaPadula
security model \cite{DBLP:journals/jar/CristiaR21} and the Tokeneer project \cite{Cristia2021}.

The empirical evaluation consists of two experiments where \setlog is asked to
determine the satisfiability of a collection of $\LCARD$ formulas. We measure
how many of those formulas are solved and the time spent in doing so. In both
experiments we use a 2 s timeout and the computing times are those of the
solved problems. The data set to reproduce these experiments can be downloaded
from \url{http://people.dmi.unipr.it/gianfranco.rossi/SETLOG/size.zip} (the
technical details can be found in \ref{app:experiments}). These experiments do
not use nested sets.

As shown in Table \ref{t:eval}, the first experiment is performed over a
collection of 468 $\LCARD$ formulas. These formulas are taken from different
sources:
\begin{itemize}
\item
\textsc{Tests}. These are simple cardinality formulas of our own.
  %
\item
\textsc{Properties}. These are formulas related to typical cardinality
properties such as $\card{A \cup B} \leq \card{A} + \card{B}$.
  %
\item
\textsc{CVC4}. These are problems used by \citeN{Bansal2018} as a benchmark for
the implementation of cardinality constraints in the CVC4 SMT solver plus
problems derived from these.
  %
\item
\textsc{Kuncak}. These are the five examples of program verification used by
\citeN{Kuncak2006} to show their algorithm that solves \textsf{BAPA} formulas.
\textsf{BAPA} is discussed in Section \ref{related}.
\item
\textsc{ssl-reachability}. This is the collection of problems used by
\citeN{Piskac2020} to evaluate their method based on a $\mathrm{LIA}^*$
encoding.
$\mathrm{LIA}^*$ is briefly discussed in Section \ref{related}.
\end{itemize}

\begin{table}
\caption{\label{t:eval}Results of the first experiment}
\begin{tabular}{lrrrrrrr}
\hline \multicolumn{1}{c}{\textsc{Collection}} &
\multicolumn{1}{c}{\textsc{\#}} & \multicolumn{2}{c}{\textsc{Satisfiable}} &
\multicolumn{2}{c}{\textsc{Unsatisfiable}} & \multicolumn{1}{c}{\textsc{\%}} &
\multicolumn{1}{c}{\textsc{Time}} \\
 &
 &
 \textsc{Slvd} &
 \textsc{Uslvd} &
 \textsc{Slvd} &
 \textsc{Uslvd} &
 & \\\hline
\textsc{Tests} & 150 & 98 & 0 & 52 & 0 & 100 & 0.5 s \\
\textsc{Properties} & 53 & 14 & 0 & 36 & 3 & 94 & 3.8 s \\
\textsc{CVC4} & 20 & 8 & 0 & 12 & 0 & 100 & 2.5 s \\
\textsc{Kuncak} & 5 & 0 & 0 & 5 & 0 & 100 & 0.0 s \\ 
\textsc{ssl-reachability} & 240 & 130 & 13 & 90 & 7 & 92 & 19.1 s \\\hline
\textsc{Total} & 468 & 250 & 13 & 195 & 10 & 95 & 25.9 s \\\hline
\end{tabular}
\end{table}

As can be seen, \setlog solves 95\% of the problems in 25.9 s, meaning an
average of 0.06 s per problem. Even if the first collection is not considered,
\setlog solves 93\% of the resulting 318 problems in 25.4 s, thus making 0.09 s
per problem. In particular, \setlog solves all the problems in the
\textsc{CVC4} and \textsc{Kuncak} collections. It also solves 92\% of the
\textsc{ssl-reachability} collection in 19.1 s (0.09 s on average) whereas
Piscak et al. manage to solve 76\% of them in 59 s (0.3 s in
average)\footnote{Piscak et al. run their evaluation on a 2018 MacBook Pro
running OS X Mojave 10.14.5 with a 2.9 GHz Intel Core i9 processor and 32GB of
RAM. Our hardware platform is older and less powerful, see below.} \cite[Table
1]{Piskac2020}. If the timeout is set to 50 s, as done by Piscak, \setlog
manages to solve 11 more problems thus solving 96\% of them (although it needs
considerably more time as some problems are solved only after several seconds).

The second experiment concerns the evaluation of $\SATCARD$ when computing
minimal solutions---cf. Section \ref{minimal} and command \verb+fix_size+ given
in Section \ref{impl}. Then, we run \setlog on the 250 satisfiable problems of
Table \ref{t:eval} that the tool is able to solve. The results are given in
Table \ref{t:evalfixsize}. This experiment sheds some light on the efficiency
of \setlog in constructing more concrete solutions of satisfiable problems. As
can be seen, \setlog is able to produce a more concrete solution to 99\% of the
satisfiable problems in 0.07 s on average. Note that the tool is not able to
find a concrete solution for three formulas whose satisfiability, nonetheless,
it was able to ascertain.

Even if the first collection of problems is removed from this experiment,
\setlog solves 99\% of the problems in 0.1 s on average.

\begin{table}
\caption{\label{t:evalfixsize}Results of the second experiment}
\begin{minipage}{.65\textwidth}
\begin{tabular}{lrrrrrrr}
\hline \multicolumn{1}{c}{\textsc{Collection}} &
\multicolumn{1}{c}{\textsc{\#}} & \multicolumn{2}{c}{\textsc{Satisfiable}} &
\multicolumn{1}{c}{\textsc{\%}} &
\multicolumn{1}{c}{\textsc{Time}} \\
 &
 &
 \textsc{Slvd} &
 \textsc{Uslvd} &
 & \\\hline
\textsc{Tests} & 98 & 97 & 1 & 99 & 0.4 s \\
\textsc{Properties} & 14 & 14 & 0 & 100 & 0.1 s \\
\textsc{CVC4} & 8 & 8 & 0 & 100 & 0.7 s \\
\textsc{ssl-reachability} & 130 & 128 & 2 & 98 & 15.3 s \\\hline \textsc{Total} & 250 & 247 & 3 & 99 & 16.5 s \\\hline
\end{tabular}
\end{minipage}
\end{table}


\subsection{\label{discussion}Discussion}

In spite of initial theoretical concerns, the empirical evaluation presented in
Section \ref{empirical} shows that, in practice, \setlog's implementation of
$\Zarba$ performs no worse than other approaches and better than special
purpose algorithms such as those by Kuncak and Piscak. It is true, however,
that in the worst case the exponential complexity of the algorithm makes it
unfit for certain problems. We can see that in the unsolved problems (23 out of 468) of Table
\ref{t:eval}.

Broadly speaking, \setlog's implementation of $\SATCARD$ goes through three
phases: \emph{a)} solve the formula with minimal concern about cardinality;
\emph{b)} compute the set of solutions of a Boolean formula derived from the
irreducible form (cf. Definition \ref{def:solved}); and \emph{c)} solve an
integer linear programming problem \emph{for each} subset of the Boolean
solutions, which presupposes the powerset of the set of Boolean solutions being
computed. Each phase of $\SATCARD$ is inherently exponential, at least, in the
worse case.

However, according to our experiments, the worst of these three problems is
\emph{c)}. Its most demanding part is not the computation of the powerset
itself but solving the integer problem for each of its elements. In fact,
\setlog uses backtracking in such a way as to avoid computing the powerset
explicitly. This problem bears some relationship with the number of set
variables of the input formula, but this is neither evident nor direct. For
example a formula such as $A_1 \cup \dots \cup A_{50} = \emptyset \land
\card{A_{43}} > 2*k + 5$ is solved in virtually no time, while a formula with
fewer variables but where $\cup$ is substituted by $\cap$ will take an
exponential time. As we have noted, the real problem is the number of solutions
returned by step \emph{b)} which determines the size of the powerset.
Unfortunately, the relationship between the input formula and the number of
solutions of the Boolean problem is complex. For example, $A_1 \cup \dots \cup
A_{50} = B$ will generate many more Boolean solutions than $A_1 \cup \dots \cup
A_{50} = B \land \bigwedge_{i=1}^{49} A_i \disj A_{i+1}$. To worsen things, if
the number of set variables is large, the integer problem to be solved for each
element of the powerset becomes increasingly more complex, consuming a non
negligible time. On the other hand, a palliative to deal with \emph{c)} is the
fact that the problem is inherently parallelizable.

The introduction of inference
rules proved to be a good method to avoid many of the exponential problems we
have discussed above. As long as the application of inference rules remains
polynomial in the size of the formula received by $\STEPCARD$, it will be, on
average, better to add them than not. It remains as an open problem whether or
not there is a set of inference rules applicable in polynomial time
constituting a decision procedure for $\LCARD$. We believe the answer is no.

\section{\label{related}Related work}

Computable Set Theory (CST) has studied the problem of deciding the
satisfiability of set formulas involving cardinality constraints since a long
time ago \cite{DBLP:conf/cade/FerroOS80}\cite[Chapter
11]{DBLP:series/mcs/CantoneOP01}. In these works cardinality formulas are
encoded as additive arithmetic formulas over the natural numbers.
\citeN{hibti1995} proves the decidability of a similar problem by encoding it
as a propositional consistency problem.

Zarba's work is rooted in CST and thus relies on the notion of \emph{place} as
a way to represent Venn regions. This notion is used only inside $\STEPCARD$.
Zarba also proves that a theory of multisets, without the cardinality operator,
is decidable \cite{DBLP:conf/cade/Zarba02}. Later on, Zarba proved that a
theory of (not necessarily finite) sets, including the cardinality operator,
combined with a theory of cardinal numbers is decidable
\cite{DBLP:journals/jar/Zarba05}.

In the field of Constraint Logic Programming a number of proposals have been
put forward introducing \emph{set constraints}, possibly including cardinality
\cite{DBLP:journals/constraints/Azevedo07,DBLP:journals/constraints/Gervet97,DBLP:journals/jair/HawkinsLS05}.
In these proposals, constraint (set) variables have a \emph{finite domain}
attached to them, which is exploited by the solver to efficiently compute
simplified forms of the original constraints or to detect failures. The same
approach is adopted in the constraint modeling language MiniZinc
\cite{minizinchandbook}. While the availability of finite domains for
constraint variables allows efficient handling of set constraints, it actually
prevents the user from using the solver as a general theorem prover. On the contrary, this is
feasible in \setlog where constraint variables do not require finite domains.
For example, proving the property $\forall A,B,n: A \subseteq B \land
\card{A}=n \land \card{B}=n \implies A = B$, can be done in \setlog by checking
that the formula
 \verb!subset(A,B) & size(A,N) & size(B,N) & A neq B!
is unsatisfiable. The same general result cannot be achieved for instance in
MiniZinc, since set variables $A$ and $B$ (declared as ``decision variables''
in MiniZinc) must have a fixed domain attached to them---e.g.,
 \verb!var set of 0..100: A!.
Thus, we can write the formula in MiniZinc but what we prove is not as general
as in \setlog: if we get an \verb!UNSATISFIABLE! answer from MiniZinc it does
not mean we have proved the (general) property, while in \setlog it does.
Furthermore, set elements in \setlog can be of any type, including unbounded
constraint variables and other sets, which are not allowed in MiniZinc and in
other related proposals for set constraints.

V. Kuncak and his colleagues have worked on the decidability of the first-order
multisorted theory \textsf{BAPA} and its applications to program verification
\cite{Kuncak2006}. \textsf{BAPA} extends the combination of the theory of
Boolean algebras of sets (BA) and Presburguer arithmetic (PA). In this way
\textsf{BAPA} can deal with formulas where the cardinality of a set is treated
as an integer variable subjected to PA constraints. Kuncak's algorithm reduces
a \textsf{BAPA} sentence to an equivalent \textsf{PA} sentence. In this way,
the algorithm enjoys several nice properties (e.g., its complexity is no worse
than an optimal algorithm for deciding \textsf{PA}). This implies that the
complexity of Kuncak's algorithm is identical to the complexity of PA. Besides,
the algorithm can eliminate quantifiers from a \textsf{BAPA} formula thus
turning this into a quantifier-free \textsf{BAPA} formula---called
\textsf{QFBAPA}. The algorithm depends upon \textsf{MAXC}, an integer constant
denoting the size of the finite universe.  Our method does not depend on any
constant denoting the size of the universe. Kuncak and his colleagues have
implemented this algorithm in the Jahob system, used to check the consistency
of data structures in the Java language. Kuncak shows a few problems related to
program verification that can be solved with his algorithm. All the problems
proposed by Kuncak can also be efficiently solved by \setlog as is shown in
Section \ref{empirical}.

In a further development, \citeN{Piskac2008} give a decision
procedure for multisets with cardinality constraints by using a similar method
(i.e., encoding  input formulas as quantifier-free PA formulas); more recently
a more efficient method based on a LIA$\star$ encoding has been proposed
\cite{Piskac2020,Levatich2020}. These algorithms have been implemented in the
MUNCH \cite{Piskac2010} and ssl-reachability \cite{Piskac2020} tools which use
existing solvers to solve the various problems involved in this approach,
e.g., linear integer arithmetic. The empirical evaluation used to evaluate the
ssl-reachability tool is included in the evaluation of the implementation of
our algorithm in \setlog (cf. Section \ref{empirical}).

\citeN{Suter2011} have extended the Z3 SMT solver to
solve problems of the \textsf{QFBAPA} logic which, as said above, can be used
to encode set problems combined with PA problems through the cardinality
operator. \citeN{Bansal2018} also approach the problem of deciding
the satisfiability of finite set formulas with cardinality in the context of
SMT solvers. They propose and implement in CVC4 a calculus describing a
combination of a procedure for reasoning about membership with a procedure for
reasoning about cardinality. Their method is based on a different strategy
w.r.t. to Suter's work but it draws the concept of \emph{place} from CST
although used in an incremental way. According to Bansal and his colleagues,
Suter's method cannot scale well when the formula has set membership
constraints because these are encoded as cardinality constraints (i.e., $x \in
A \iff \{x\} \subseteq A$ and $\{x\}$ is actually a set whose cardinality is
1). Instead, they propose to avoid dealing with set membership constraints in
terms of \emph{places} or Venn regions, but to reason directly about
membership. This is aligned with how our method deals with set membership,
although we do it in terms of set unification \cite{Dovier2006}. In fact, in our
method a formula such as $x \in B \cup C$ is written as $\Cup(B,C,A) \land x
\in A$ which in turn is rewritten as $A = \{x \plus N\} \land \Cup(B,C,\{x
\plus N\})$, where $N$ is a new variable (implicitly existentially quantified)
and $\{x \plus N\}$ is a set constructor interpreted as $\{x\} \cup N$. No Venn
regions are computed when this formula is solved. Bansal et al. empirically
evaluate their method on 25 problems on program verification. The first 15 of
these problems are drawn from the evaluations performed by Kuncak and Suter on
their tools. CVC4 shows a comparable performance w.r.t. those other tools.
These 15 problems are included in the empirical evaluation of our method
reported in Section \ref{empirical}; \setlog also shows a comparable
performance. Bansal et al. also compare their method with Suter's on the
constraint $x \in A_1 \cup \dots \cup A_{21}$. As expected, Suter's method runs
out of memory after some time while CVC4 solves the formula immediately.
\setlog also solves the formula quickly and is able to return a finite
representation of all possible solutions which, as far as we know, no other
tool can do. \setlog also supports nested sets which is apparently not the case
of CVC4.

\citeN{DBLP:conf/vmcai/YessenovPK10} prove the
decidability of a theory of sets including functions, $n$-ary relations and
some operators for the algebra of relations (e.g., relational image). Then, they
show that the cardinality operator can be added to the theory preserving its
decidability.

\citeN{DBLP:journals/constraints/Azevedo07} describes the
\textit{Cardinal} system which is part of the ECLiPSe Prolog library.
\textit{Cardinal} is based on constraint propagation on set cardinality and set
interval reasoning. Methods of this kind are, in general, restricted to
formulas where  the
cardinality of each set is constrained to range over a closed integer interval.
Azevedo applies his method to some problems on digital circuits.

A proposal for extending \setlog with integers and cardinality constraints had
already been put forward in a previous work \cite{Palu:2003:IFD:888251.888272}.
In that case, however, the extension is based on the integration of CLP(FD)
into \setlog. Consequently, completeness of the solver is obtained only if
finite domains are provided for all integer variables and labelling is
performed over them. This in fact implies an upper limit for set cardinalities.
Furthermore, the presence of labeling can easily lead to unacceptable
performance.

\citeN{Alberti2017} extend linear integer arithmetic with free function symbols
and cardinality constraints for interpreted sets. Interpreted sets are sets of
the form $\{x \in [0,N)| \varphi\}$, for some $0 < N \in \nat$, and $\varphi$
is an arithmetic formula. Free unary function symbols are used to represent
array ID's. Thus, the language offers terms of the form $a(y)$ where $a$ is an
array ID and $y$ is a variable. Formulas such as $a(y) < 1$ are allowed to
occur in interpreted sets where $y$ is the bound variable. Then, the language
\emph{only} allows one to indicate the cardinality of interpreted sets, e.g.,
$\card{\{y \in [0,N) | a(y) < 1\}} = 0$. These authors prove that some
fragments of this logic are both decidable and expressive enough as to model
and reason about problems of fault-tolerant distributed systems. The
decidability results are obtained by mapping those fragments into Presburger
arithmetic enriched with unary counting quantifiers. One of the decidable
fragments has been implemented in a tool that uses the Z3 SMT solver as a
back-end solver for quantifier-free linear arithmetic. Alberti's logic does not
include classic set theoretic operators such as union. Hence, it is difficult
to compare the expressiveness of Alberti's logic with other logics analyzed in
this section and with ours. Although \setlog's intensional sets
\cite{DBLP:journals/jar/CristiaR21a} could be used to encode Alberti's
interpreted sets, it is still necessary to extend that theory as to compute the
cardinality of intensional sets. This is a line of future research.

\citeN{Bender2017} extend some of the previous
results to theories where cardinalities are replaced by the more general notion
of measures. In this case a key aspect of the previous approaches is no longer
valid, namely the fact that only the empty set has  cardinality equal to 0, as
there are non-empty sets with measure 0. The theories analyzed by these authors
are important in, for example, duration calculus.

Also the Artificial Intelligence community has studied the problem of reasoning
about the size of sets, e.g.,
\cite{ding_harrison-trainor_holliday_2020,DBLP:conf/aaai/KisbyBKM20}.  We want
to remark the work by \citeN{DBLP:conf/aaai/KisbyBKM20} because
they propose two logics, combining sets with cardinality, whose decidability
can be solved in polynomial time. As expected, the gain in complexity is at the
cost of expressiveness. Nonetheless, the result may deserve being studied in
terms of software verification as it might give clues about what are the
simplest specifications and proof obligations involving sets and cardinality.
From there, compositional methods might be drawn in order to tame the
complexity constantly faced in automated program verification.

\section{\label{concl}Concluding Remarks}

In this paper we have presented a decision
procedure for the algebra of hereditarily finite hybrid sets extended with
cardinality constraints. The proposed
procedure is implemented within \setlog, a CLP system able to deal with a few
decidable fragments of set theory. The empirical evaluation carried out on the
implementation proves that \setlog is able to deal efficiently with formal
verification problems involving cardinality constraints.

As a future work, we plan to use this decision procedure as the base for a
decision procedure for the algebra of finite sets extended with integer
intervals. Indeed, the following identity:
\begin{equation*}
A = [m,n] \iff A \subseteq [m,n] \land \card{A} = n-m+1
\end{equation*}
becomes the key for a set unification algorithm including integer intervals
with \emph{variable} limits. In fact, it would suffice to be able to deal with
constraints of the form $A \subseteq [m,n]$ in a decidable framework to have a
decision procedure for integer intervals. In turn, integer intervals are a key
component in the definition of arrays as sets. In fact, if $array(A,n)$ is a
predicate stating that $A$ is an  array of length $n$ whose components take
values on some universe $\mathcal{U}$, then it can be defined as follows:
\begin{equation*}
array(A,n) \iff A: [1,n] \fun \mathcal{U}
\end{equation*}
\setlog already supports a broad class of set relation algebras
(\citeANP{DBLP:journals/jar/CristiaR20}
\citeyearNP{DBLP:journals/jar/CristiaR20,DBLP:conf/RelMiCS/CristiaR18}),
including partial functions and the domain operator. Hence, it would be
possible to use \setlog to automatically reason about broad classes of programs
with arrays from a set theoretic perspective which would be different from
existing approaches
\cite{DBLP:conf/lics/StumpBDL01,DBLP:conf/vmcai/BradleyMS06}.

\bigskip

\noindent\textit{Competing interests: The authors declare none}

\bigskip

\bibliographystyle{acmtrans}
\bibliography{/home/mcristia/escritos/biblio}

\appendix


\section{Proofs}\label{proofs}

In this section we provide the proofs of equisatisfiability of the main rewrite
rules for the $\Size$ constraint. Note that the equisatisfiability property
for rule \eqref{size:empty} and for rule \eqref{size:zero} is trivial. Then we
give the proofs for rule \eqref{size:ext} and \eqref{size:const3}.

\begin{lemma}[Equisatisfiability of rule \eqref{size:ext}]
\begin{flalign*}
& \forall x, A, m: \\
& \t1 \Size(\{x \plus A\},m) \iff \\
&  \t2 \exists n: x \notin A \land m = 1 + n \land \Size(A,n) \\
&  \t2 \lor \exists N: A = \{x \plus N\} \land x \notin N \land \Size(N,m)
\end{flalign*}
\end{lemma}

\begin{proof}
First, assume $x \notin A$.
\begin{flalign*}
& \Size(\{x \plus A\},m) \\
& \iff \card{\{x \plus A\}} = m \why{semantics of $\Size$} \\
& \iff \card{\{x\}\cup A} = m
     \why{semantics of $\{\cdot\plus\cdot\}$} \\
& \iff \card{\{x\}} + \card{A} = m
     \why{$x \notin A$ and property $\card{\cdot}$} \\
& \iff 1 + \card{A} = m \why{property of $\card{\cdot}$} \\
& \iff 1 + n = m \land n = \card{A} \why{substitution} \\
& \iff 1 + n = m \land \Size(A,n) \why{semantics of $\Size$}
\end{flalign*}

Now, assume $x \in A$. Then, take $N = A \setminus \{x\}$. Trivially, $A =
\{x\} \cup N$ and $x \notin N$. Now, $A = \{x \plus N\}$ \by{semantics of
$\{\cdot\plus\cdot\}$}. Finally:
\begin{flalign*}
& \Size(\{x \plus A\},m) \\
& \iff \card{\{x \plus A\}} = m \why{semantics of $\Size$} \\
& \iff \card{\{x\}\cup A} = m
     \why{semantics of $\{\cdot\plus\cdot\}$} \\
& \iff \card{A} = m
     \why{$x \in A \implies \{x\} \cup A = A$} \\
& \iff \Size(A,m) \why{semantics of $\Size$}
\end{flalign*}

And this finishes the proof.
\end{proof}

\begin{lemma}[Equisatisfiability of rule \eqref{size:const3}]
\begin{flalign*}
& \forall A, c: c > 0 \implies \\
& \t1 \Size(A,c) \iff \exists y_1,\dots,y_c: A = \{y_1,\dots,y_c\}
  \land ad(y_1,\dots,y_c)
\end{flalign*}
where:
\[
ad(y_1,\dots,y_c) \defs \bigwedge_{i=1}^{c-1}
\bigwedge_{j=i+1}^c y_i \neq y_j
\]
\end{lemma}

\begin{proof}
\begin{flalign*}
& \Size(A,c) \\
& \iff \card{A} = c \why{semantics of $\Size$} \\
& \iff A = \{y_1,\dots,y_c\} \land ad(y_1,\dots,y_c)
     \why{semantics of $\card{\cdot}$ and $c > 0$}
\end{flalign*}
for some elements $y_1,\dots,y_c$.
\end{proof}


\section{Mapping $\LCARD$ Formulas into $\LZa$ Formulas}\label{mapping}

In this section we define a mapping of $\LCARD$ formulas into $\LZa$
formulas. Actually, in order to justify Theorem \ref{satisf}, we only need to
map the $\LCARD$ formulas in irreducible form that are passed in to $\Zarba$.
Indeed, the implementation of $\Zarba$ is called on $\LCARD$ formulas in
irreducible form, as explained in Section \ref{sizesolver}.

Hence, we define a function, $\map$, that takes $\LCARD$ terms, constraints or
formulas in irreducible form and returns $\LZa$ terms, constraints or formulas.

\paragraph{Variables.}
Variables are mapped onto themselves taking care of their sort:
\begin{equation*}
\map(x) \defs x \text{, if $x \in \Var$}
\end{equation*}

\paragraph{Ur-elements.}
Ur-elements are mapped onto themselves:
\begin{equation*}
\map(x) \defs x \text{, if $x$ is of sort $\Ur$}
\end{equation*}

\paragraph{Integer terms.}
As $\LZa$ only provides the constants 0 and 1, the mapping of $n \in \num$ is
as follows:
\begin{flalign*}
  \quad\quad & \map(0) \defs 0 & \\
  & \map(n) \defs \overbrace{1+\dots+1}^n = \sum_{i=1}^n 1 \text{, for $n \neq
  0$} &
\end{flalign*}

$\LZa$ does not provide the integer product. However, recall that $\LCARD$
admits only linear terms so in $n*m$ at least one is a constant; if it is $m$,
then we first switch the term as $m*n$. In this case the mapping for integer
linear terms is as follows:
\begin{flalign*}
  \quad\quad & \map(-m) \defs -\map(m) & \\
 & \map(n + m) \defs \map(n) + \map(m) & \\
 & \map(n - m) \defs \map(n) - \map(m) & \\
 & \map(n * m) \defs \overbrace{\map(m) + \dots + \map(m)}^n = \sum_{i=1}^n
\map(m) &
\end{flalign*}

\paragraph{Integer constraints.}
\begin{flalign*}
  \quad\quad & \map(n = m) \defs \map(n) = \map(m) & \\
 & \map(n \leq m) \defs \map(n) < \map(m) \lor \map(n) = \map(m) &
\end{flalign*}

\paragraph{Set terms.}
Recall that we only need to map set terms in irreducible form
except those at the right of an equality of the form $\dot{X} = t$. This means
that, actually, we do not need to map any set term.

\paragraph{Set constraints.}
Again, we only need to map set constraints appearing in irreducible form.
Moreover, we do not need to map constraints based on $=$, $\notin$ and $\neq$,
as explained in Section \ref{integrating}. Therefore, we only need to map
constraints based on $\Cup$, $\disj$ and $\Size$.
\begin{flalign*}
  \quad\quad & \map(\Cup(A,B,C)) \defs \map(C) = \map(A) \cup \map(B) & \\
 & \map(A \disj B) \defs \map(A) \cap \map(B) = \emptyset & \\
 & \map(\Size(A,K)) \defs \card{\map(A)} = \map(K) &
\end{flalign*}

\paragraph{Formulas.}
The irreducible form is a conjunction of constraints in irreducible form. Then,
we only need to map conjunctions of constraints.
\begin{equation*}
\map(p \land q) \defs \map(p) \land \map(q)
\end{equation*}


\section{\label{app:datacontainer}A Simple \setlog Program}

The following \setlog program models a simple data container and its cache.  As
long as the container \verb+Cont+ holds at most \verb+N+ elements its cache
\verb+Cache+ holds the same elements; when \verb+Cont+ grows beyond \verb+N+,
\verb+Cache+ contains only \verb+N+ elements. In this model, both \verb+Cont+
and \verb+Cache+ are sets.
\begin{verbatim}
cache(Cont,N,Cache) :-
  0 < N &
  size(Cont,S) &
  (S =< N &
   Cache = Cont
   or
   S > N &
   un(Rest,Cache,Cont) &
   disj(Rest,Cache) &
   size(Cache,N)
  ).
\end{verbatim}

In this way, we can run queries to play with \verb+cache+:
\begin{verbatim}
{log}=> cache({1,b,[2,q]},2,Cache).

Cache = {b,[2,q]}

Another solution?  (y/n)
Cache = {1,[2,q]}

Another solution?  (y/n)
Cache = {1,b}

Another solution?  (y/n)
no
\end{verbatim}
Given that \verb+Cont+ and \verb+Cache+ are sets, \verb+cache+ returns several
solutions where \verb+Cache+ holds different elements of \verb+Cont+. In other
words, this model of the system is non-deterministic as we cannot say what are
the first elements to be put in the cache. Determinism can be imposed by
calling \verb+cache+ in this way:
\begin{verbatim}
{log}=> cache({1,b,[2,q]},2,C)!.

C = {b,[2,q]}

Another solution?  (y/n)
no
\end{verbatim}

\setlog can be used to prove that \verb+cache+ verifies some properties.  For
example, if \verb+M+ is the size of \verb+Cont+ and we have that \verb+N < M+
then \verb+Cache+ is a non-empty set. This is proved by running a query
representing the negation of this property:
\begin{verbatim}
{log}=> cache(Cont,N,Cache) & size(Cont,M) & N < M & Cache = {}.
\end{verbatim}
In which case \setlog answers \verb+no+ meaning the query cannot be satisfied.


\section{\label{app:experiments}Technical details of the empirical evaluation}

The experiments described in Section \ref{empirical} were performed on a
Latitude E7470 (06DC) with a 4 core Intel(R) Core\texttrademark{} i7-6600U CPU
at 2.60GHz with 8 Gb of main memory, running Linux Ubuntu 18.04.5 (LTS) 64-bit
with kernel 4.15.0-135-generic. \setlog 4.9.8-7g over SWI-Prolog
(multi-threaded, 64 bits, version 7.6.4) was used during the experiments.

Each \setlog formula was run within the following Prolog program:
\begin{verbatim}
   consult('setlog.pl').
   set_prolog_flag(answer_write_options,[max_depth(0)]).
   set_prolog_flag(toplevel_print_options,
                   [quoted(true),
                    portray(true), spacing(next_argument)]).
   time(once(rsetlog(<FORMULA>), 2000,__C,__R,[]))).
\end{verbatim}
where \verb+<FORMULA>+ is replaced by each formula, \verb+2000+ is the timeout
(in milliseconds), and \verb+__C+ and \verb+__R+ are used to get the result of
the execution. Each of these programs was run from the command line as follows:
\begin{verbatim}
   prolog -q < <PROG>
\end{verbatim}
The execution time is the one printed by the \verb+time/1+ predicate.


\section{Inequality elimination ($\mathsf{remove\_neq}$)}\label{inequalities}

The $\CARD$-formula returned by Algorithm \ref{glob} when $\mathsf{STEP_S}$
reaches a fixpoint is not necessarily satisfiable.

\begin{example} [Unsatisfiable formula returned by
$\mathsf{STEP_S}$]\label{ex:cupcup} The $\CARD$-formula:
\begin{equation}\label{eq:ex1}
\Cup(A,B,C) \land \Cup(A,B,D) \land C \neq D
\end{equation}
cannot be further rewritten by any of the rewrite rules considered above.
Nevertheless, it is clearly unsatisfiable. \qed
\end{example}

In order to guarantee that $\SATCARD$ returns either $\false$ or satisfiable
formulas (see Theorem \ref{satisf}), we still need to remove all inequalities
of the form $\dot{A} \neq t$, where $\dot{A}$ is of sort $\sSet$, occurring as
an argument of $\CARD$-constraints based on $\Cup$ or $\Size$. This is
performed (see Algorithm \ref{glob}) by executing the routine
\textsf{remove\_neq}, which applies the rewrite rule described by the generic
rule scheme of Figure \ref{fig:rules_neq_elim}. Basically, this rule exploits
set extensionality to state that two sets that differ can be distinguished by
asserting that a fresh element ($\dot n$) belongs to one but not to the other.
Notice that the third disjunct is necessary when $t$ is a non-set term. In this
case the second disjunct is false while the first disjunct forces $\dot{A}$ to
contain an element $\dot{n}$; so without the third disjunct we would miss the
solution $\dot{A} = \e$.

\begin{figure}
\hrule\vspace{3mm}
 \raggedright
 \quad\quad If $A \in \Var_\Set$; $t : \langle
 \{\sSet,\sUr\} \rangle$; $\Phi$ is the input formula then:
\begin{flalign*}
 \quad\quad \quad\quad & \text{If $\dot{A}$ occurs as an argument of a $\pi$-constraint,
    $\pi \in \{\Cup, \Size\}$, in $\Phi$:} & \\
 & \dot{A} \neq t \lfun (\dot{n} \in \dot{A} \land \dot{n} \notin t) \lor
    (\dot{n} \in t \land \dot{n} \notin \dot{A}) \lor (\dot{A} = \e \land t \neq
 \e) &
 \label{rule:neq_elim}
\end{flalign*}
\hrule
 \caption{Rule scheme for $\neq$ constraint elimination rules}
\label{fig:rules_neq_elim}
\end{figure}

\begin{example}[Elimination of $\neq$ constraints]
The $\CARD$-formula of Example \ref{ex:cupcup} is rewritten to (we do not
consider the third disjunct as $C$ and $D$ are set variables):
\begin{flalign*}
\quad\quad & \Cup(A,B,C) \land \Cup(A,B,D) \land C \neq D \lfun & \\
 & \Cup(A,B,C)
  \land \Cup(A,B,D)
  \land (\dot{n} \in C \land \dot{n} \notin D \lor \dot{n} \notin C \land \dot{n} \in D)
  \lfun & \\
 & \Cup(A,B,C)
  \land \Cup(A,B,D)
  \land \dot{n} \in C \land \dot{n} \notin D & \\
 & {}\lor & \\
 & \Cup(A,B,C)
  \land \Cup(A,B,D) \land \dot{n} \notin C \land \dot{n} \in D &
\end{flalign*}
Then, the $\in$ constraint in the first disjunct is rewritten into a $=$
constraint (namely, $C = \{\dot{n} \plus \dot{N}\}$), which in turn is
substituted into the $\Cup$ constraints, which in turn are further rewritten by
rules such as those shown in Figure \ref{f:clpset} and \cite{Dovier00}. This
process will eventually return $\false$, at which point the second disjunct is
processed in a similar way. \qed
\end{example}

\end{document}